\documentclass[english,twocolumn,superscriptaddress,letterpaper]{revtex4}

\usepackage{ae} 
\usepackage[T1]{fontenc}
\usepackage[ansinew]{inputenc}
\usepackage{amsmath,amsthm}
\usepackage{amssymb}
\usepackage{dsfont}
\usepackage{graphicx}
\usepackage{color}
\usepackage{epsfig}
\usepackage{booktabs}
\usepackage{makecell}

\numberwithin{equation}{section}
\renewcommand{\theequation}{\arabic{section}.\arabic{equation}}

\def\d{{\rm d}}
\newtheorem{theorem}{Theorem}

\usepackage{varioref}
\usepackage{makeidx}
\makeindex

\usepackage[english]{babel}

\begin{document}
\phantom{.}
\newpage
\hfill \hbox{MITP/17-097}

\title{Differential equations for loop integrals in Baikov representation}

\author{Jorrit Bosma}
\affiliation{ETH Z{\"u}rich, Wolfang-Pauli-Strasse 27, 8093 Z{\"u}rich, Switzerland}\email{jbosma@itp.phys.ethz.ch}
\author{Kasper J. Larsen}
\affiliation{School of Physics and Astronomy, University of Southampton, Highfield, Southampton, SO17 1BJ, United Kingdom}\email{Kasper.Larsen@soton.ac.uk}
\author{Yang Zhang}
\affiliation{ETH Z{\"u}rich, Wolfang-Pauli-Strasse 27, 8093 Z{\"u}rich, Switzerland}
\affiliation{PRISMA Cluster of Excellence, Johannes Gutenberg University, 55128 Mainz, Germany}\email{zhang@uni-mainz.de}

\begin{abstract}
We present a proof that differential equations for Feynman loop integrals
can always be derived in Baikov representation without involving dimension-shift identities.
We moreover show that in a large class of two- and three-loop diagrams it is possible to avoid
squared propagators in the intermediate steps of setting up the differential equations.
\end{abstract}

\maketitle

\section{Introduction}\label{sec:introduction}

The physics program of the Large Hadron Collider (LHC) demands precision
calculations of production cross sections in order to have accurate predictions of the
Standard Model signals and background. To match the experimental precision
and the parton distribution function uncertainties, this requires for most processes
computations at next-to-next-to leading order (NNLO) in fixed-order perturbation theory.

Out of the contributions to NNLO cross sections, the double-virtual one,
i.e.~the two-loop scattering amplitude, is the main theoretical bottleneck.
Two-loop amplitudes for most $2 \to 2$ processes of
phenomenological interest have now been computed%
~\cite{Bonciani:2008az,Gehrmann:2009vu,Bonciani:2009nb,Bonciani:2010mn,Gehrmann:2011ab,%
Gehrmann:2011aa,Gehrmann:2013vga,vonManteuffel:2013uoa,Gehrmann:2013cxs,Henn:2013woa,%
Bonciani:2013ywa,Henn:2014lfa,Gehrmann:2014bfa,Caola:2014lpa,Caola:2014iua,Papadopoulos:2014hla,%
Huber:2015bva,Gehrmann:2015ora,Caola:2015ila,vonManteuffel:2015msa,Bonciani:2016ypc,%
Bonciani:2016qxi,Melnikov:2017pgf,Mastrolia:2017pfy,Becchetti:2017abb}. Results for
$2 \to n$ processes with $n > 2$ are largely restricted to all-plus helicity
amplitudes~\cite{Badger:2013gxa,Badger:2015lda,Gehrmann:2015bfy,%
Dunbar:2016aux,Dunbar:2016cxp,Dunbar:2016gjb,Badger:2016ozq,Dunbar:2017nfy},
though impressive results on all helicity amplitudes for the planar five-gluon
amplitude have recently appeared \cite{Badger:2017jhb}. Expressions for the
planar five-gluon basis integrals were computed in refs.~\cite{Gehrmann:2015bfy,Papadopoulos:2015jft}.

Calculation of multi-loop amplitudes proceeds in two stages.
In the first step, the expression for the amplitude
obtained from the Feynman rules after tensor reduction
is recast as a linear combination of a basis of integrals, which leads
to a significant reduction in the number of contributing integrals.
(The existence of a finite basis was proved in ref.~\cite{Smirnov:2010hn}.)
This step is performed by use of integration-by-parts (IBP) reductions.
These arise from the vanishing integration of total derivatives
in dimensional regularization,
\begin{equation}
\int \prod_{i=1}^L \frac{\d^D \ell_i}{\mathrm{i} \pi^{D/2}}
\sum_{j=1}^L \frac{\partial}{\partial \ell_j^\mu}
\frac{v_j^\mu \hspace{0.5mm} P}{D_1^{\alpha_1} \cdots D_m^{\alpha_m}}
\hspace{1mm}=\hspace{1mm} 0 \,, \label{eq:IBP_schematic}
\end{equation}
where $P$ and the vectors $v_j^\mu$ are polynomial in the internal and
external momenta, the $D_k$ denote inverse propagators, and the
$\alpha_i$ are integers. In practice, the IBP identities
generate a large set of linear relations between loop integrals
which allows most of them to be reexpressed in terms of a set of
basis integrals. The step of solving the linear systems that arise from
eq.~(\ref{eq:IBP_schematic}) may be carried out by Gaussian elimination
in the form of the Laporta algorithm~\cite{Laporta:2000dc,Laporta:2001dd},
leading in general to relations that involve integrals with
squared propagators. There are several publicly available implementations
of automated IBP reduction: AIR~\cite{Anastasiou:2004vj},
FIRE~\cite{Smirnov:2008iw,Smirnov:2014hma},
Reduze~\cite{Studerus:2009ye,vonManteuffel:2012np},
LiteRed~\cite{Lee:2012cn}, Kira~\cite{Maierhoefer:2017hyi}, along with private implementations.
A formalism for deriving IBP reductions without squared
propagators was developed in refs.~\cite{Gluza:2010ws,Ita:2015tya}.
A systematic way of deriving IBP reductions on
generalized-unitarity cuts was given in ref.~\cite{Larsen:2015ped}.
A recent approach~\cite{vonManteuffel:2014ixa} uses sampling
of finite-field elements to construct the reduction coefficients.

In the second step of computing a multi-loop amplitude, one sets up
differential equations satisfied by the basis integrals%
~\cite{Kotikov:1990kg,Kotikov:1991pm,Bern:1993kr,Remiddi:1997ny,Gehrmann:1999as,Henn:2013pwa,%
Papadopoulos:2014lla,Ablinger:2015tua,Liu:2017jxz}. Letting $x_m$
denote an external kinematical invariant, $\epsilon = \frac{4-D}{2}$ the
dimensional regulator, and $\boldsymbol{\mathcal{I}}(\boldsymbol{x},\epsilon)
= (\mathcal{I}_1(\boldsymbol{x},\epsilon), \ldots, \mathcal{I}_N(\boldsymbol{x},\epsilon))$ the basis
of integrals, the result of differentiating any basis integral
with respect to $x_m$ can be expressed as a linear combination
of the basis integrals, once again by use of the IBP reductions. In this way,
one finds a first-order linear system of partial differential equations,
\begin{equation}
\frac{\partial}{\partial x_m} \boldsymbol{\mathcal{I}}(\boldsymbol{x},\epsilon)
= A_m (\boldsymbol{x}, \epsilon) \boldsymbol{\mathcal{I}}(\boldsymbol{x},\epsilon) \,.
\label{eq:diff_eqs_schematic}
\end{equation}
The boundary conditions can be fixed by using known results for a
(typically small) subset of the simpler basis integrals
and furthermore imposing regularity conditions, such as
requiring vanishing discontinuities in certain kinematic channels.
Having determined appropriate boundary conditions,
the differential equations (\ref{eq:diff_eqs_schematic}) can
be solved to produce expressions for the basis integrals.
In practice, this method has proven to be a powerful tool
for computing multi-loop integrals, enabling for example
the calculation of the basis integrals for the $2 \to 2$
two-loop amplitudes mentioned above.

An important insight for the latter calculations was
the observation \cite{Henn:2013pwa} that in many cases
(namely those where the basis integrals evaluate into
multiple polylogarithms), the basis can be chosen such that
the $\epsilon$-dependence factors out of the coefficient
matrix, $A_m (\boldsymbol{x}, \epsilon) = \epsilon \widehat{A}_m (\boldsymbol{x})$.
In such a basis, often referred to as a canonical basis,
integrating the differential equations in a Laurent expansion
around $\epsilon = 0$ becomes trivial.

In this paper we study differential equations of the type
in eq.~(\ref{eq:diff_eqs_schematic}) in the context of the IBP reduction formalisms
of refs.~\cite{Gluza:2010ws,Ita:2015tya,Larsen:2015ped,Bern:2017gdk}.
The main idea of the latter is to find appropriate choices of $v_j^\mu (\ell_i)$
in eq.~(\ref{eq:IBP_schematic}) for which the resulting IBP identities
do not involve squared propagators. These IBP identities have
the advantage of involving a more limited set of integrals than
generic IBP identities, and as a result produce smaller linear systems
to be solved.

The aim of this paper is two-fold. First, we prove
that the differential equations (\ref{eq:diff_eqs_schematic}) can be derived in Baikov
representation, used in several recent papers \cite{Larsen:2015ped,Frellesvig:2017aai,Zeng:2017ipr,Harley:2017qut},
without involving dimension-shift identities. This is advantageous
because the latter are computationally intensive to generate.
Second, we investigate to what extent the IBP reduction formalism
discussed in the previous paragraph is compatible
with the differential equations (\ref{eq:diff_eqs_schematic}).
We will show that, for a large class of diagrams,
it is possible to avoid having the derivatives on the left-hand side
of eq.~(\ref{eq:diff_eqs_schematic}) produce integrals
with squared propagators. The obtained differential equations
can subsequently be put in a canonical form
$A_m (\boldsymbol{x}, \epsilon) = \epsilon \widehat{A}_m (\boldsymbol{x})$
by applying the Lee-Moser algorithm \cite{Moser:1959,Lee:2014ioa,Lee:2017oca}
in single-scale cases, or multi-scale generalizations \cite{Meyer:2016slj}.
The algorithms find a change-of-basis matrix to a canonical basis.
They have been implemented in the publicly available codes \texttt{epsilon}~\cite{Prausa:2017ltv}
and \texttt{Fuchsia}~\cite{Gituliar:2017vzm}, and \textit{CANONICA}~\cite{Meyer:2017joq},
respectively.

This paper is organized as follows.
In Sec.~\ref{sec:Baikov_representation}
we set up notation and give the Baikov representation
of a generic loop integral. In Sec.~\ref{sec:diff_eqs_in_Baikov_representation}
we derive differential equations in Baikov representation and
show that dimension shifts can be avoided provided that
a certain ideal membership condition holds.
In Sec.~\ref{sec:proof_of_fundamental_ideal_membership} we
give a general proof of the latter ideal membership.
In Sec.~\ref{sec:diff_eqs_without_squared_propagators}
we show that if a stronger ideal membership holds,
then integrals with squared propagators can be avoided
in intermediate stages. In Sec.~\ref{sec:examples}
we give several examples where the stronger ideal
membership is found to hold, along with the explicit
differential equations. We also give a counterexample.
In Sec.~\ref{sec:conclusions} we give our conclusions.
In Appendix~\ref{sec:diff_eqs_via_dimension_shifts} we
derive differential equations in Baikov representation
via dimension-shift identities.

\section{Baikov representation}\label{sec:Baikov_representation}

In this section we set up notation and introduce the Baikov
representation of a generic Feynman loop integral.

We consider an $L$-loop integral with $k$ propagators and
$m-k$ irreducible scalar products (i.e., polynomials in the
loop momenta and external momenta which cannot be expressed
as a linear combination of the inverse propagators).
We work in dimensional regularization and normalize the
integral as follows,
\begin{equation}
I(N; \alpha_1, \ldots, \alpha_m; D) \equiv \int \prod_{j=1}^L \frac{\d^D \ell_j}{\mathrm{i} \pi^{D/2}}
\frac{N}{D_1^{\alpha_1} \cdots D_m^{\alpha_m}} \,. \label{eq:def_generic_Feynman_integral}
\end{equation}
Here $N$ denotes a polynomial in the external momenta $p_1, \ldots, p_E, p_{E+1}$
and the loop momenta $\ell_1, \ldots, \ell_L$. The propagators are labeled so that
\begin{align}
\alpha_i & \geq 1 \hspace{5mm} \mathrm{for} \hspace{5mm} i=1,\ldots, k \nonumber \\
\alpha_i & \leq 0 \hspace{5mm} \mathrm{for} \hspace{5mm} i=k+1,\ldots, m \,.
\label{eq:labeling_scheme_for_propagators_and_ISPs}
\end{align}
We remark that the notation (\ref{eq:def_generic_Feynman_integral}) does not give
a unique representation, since $D_{k+1}^{-\alpha_{k+1}} \cdots D_m^{-\alpha_m}$ can also be absorbed into $N$
to form a polynomial numerator,
\begin{equation}
I(N; \alpha; D) = I\Big( N \hspace{-1mm} \prod_{j=k+1}^m \hspace{-1mm} D_j^{-\alpha_j}; \hspace{0.5mm}
(\alpha_1, \ldots, \alpha_k, \boldsymbol{0}); D\Big) \,.
\end{equation}
Nevertheless, we find it more convenient to use this notation than to fix
the projective invariance.

Our aim is now to present the Baikov representation \cite{Baikov:1996rk} of
the integral (\ref{eq:def_generic_Feynman_integral}). To this end, we start
by taking the set of all independent external and loop momenta,
\begin{equation}
\{ v_1, \ldots, v_{E+L} \} = \{ p_1, \ldots, p_E, \ell_1, \ldots, \ell_L \} \,,
\end{equation}
and form their Gram matrix $S$,
\begin{equation}
S =
\left(
 \begin{array}{ccc|ccc}
  x_{1,1}  &  \hspace{-1mm} \cdots \hspace{-1mm}  &  \hspace{-1mm} x_{1,E}    &  x_{1,E+1}    &  \hspace{-1mm} \cdots \hspace{-1mm}  &  \hspace{-1mm} x_{1,E+L} \\
  \vdots   &  \hspace{-1mm} \ddots \hspace{-1mm}  &  \hspace{-1mm} \vdots     &  \vdots       &  \hspace{-1mm} \ddots \hspace{-1mm}  &  \hspace{-1mm} \vdots  \\
  x_{E,1}  &  \hspace{-1mm} \cdots \hspace{-1mm}  &  \hspace{-1mm} x_{E,E}    &  x_{E,E+1}    &  \hspace{-1mm} \cdots \hspace{-1mm}  &  \hspace{-1mm} x_{E,E+L} \\
\hline
 x_{E+1,1} &  \hspace{-1mm} \cdots \hspace{-1mm}  &  \hspace{-1mm} x_{E+1,E}  &  x_{E+1,E+1}  &  \hspace{-1mm} \cdots \hspace{-1mm}  &  \hspace{-1mm} x_{E+1,E+L} \\
 \vdots    &  \hspace{-1mm} \ddots \hspace{-1mm}  &  \hspace{-1mm} \vdots     &  \vdots       &  \hspace{-1mm} \ddots \hspace{-1mm}  &  \hspace{-1mm} \vdots  \\
 x_{E+L,1} &  \hspace{-1mm} \cdots \hspace{-1mm}  &  \hspace{-1mm} x_{E+L,E}  &  x_{E+L,E+1}  &  \hspace{-1mm} \cdots \hspace{-1mm}  &  \hspace{-1mm} x_{E+L,E+L}
\end{array} \right) \,,
\label{eq:extended_Gram_matrix}
\end{equation}
where,
\begin{equation}
x_{ij} = v_i \cdot v_j \,.
\end{equation}
It will be convenient for the derivations to follow
to relabel the entries of the upper-left $E \times E$ block
of $S$,
\begin{equation}
\lambda_{ij} = x_{ij} \hspace{5mm} \mathrm{for} \hspace{5mm} 1 \leq i, j \leq E \,,
\label{eq:definition_of_lambda}
\end{equation}
to emphasize that they are formed out of the external momenta only.
Furthermore, it is useful to define the Gram matrix $G$ of
the independent external momenta,
\begin{equation}
G =
\begin{pmatrix}
\lambda_{1,1} & \cdots & \lambda_{1,E} \\
\vdots & \ddots & \vdots \\
\lambda_{E,1} & \cdots & \lambda_{E,E}
\end{pmatrix} \,.
\label{eq:Gram_matrix_of_external_momenta}
\end{equation}
The entries of the remaining blocks of $S$ involve the loop momenta.
Because $S$ is a symmetric matrix, the entries of the upper-right $E \times L$ block
along with the upper-triangular entries of the lower-right $L \times L$ block,
\begin{equation}
x_{ij} \hspace{3mm} \mathrm{where} \hspace{2mm}
\left\{ \begin{array}{l}
1 \leq i \leq E \hspace{3mm} \mathrm{and} \hspace{3mm} E{+}1 \leq j \leq E{+}L \,, \\[2mm]
E{+}1 \leq i \leq j \leq E{+}L \,,
\end{array} \right.
\label{eq:independent_loop_containing_entries}
\end{equation}
are independent. As each set contributes respectively $LE$ and $\frac{L(L+1)}{2}$
elements, we find that the number of independent entries of $S$ that depend on the
loop momenta is $LE$ + $\frac{L(L+1)}{2}$. Since any inverse propagator $D_\alpha$ can be written
as a unique linear combination of the $x_{ij}$, we can thus conclude that the
combined number of propagators and irreducible scalar products in
eqs.~(\ref{eq:def_generic_Feynman_integral})--(\ref{eq:labeling_scheme_for_propagators_and_ISPs})
is given by,
\begin{equation}
m = LE + \frac{L(L+1)}{2} \,.
\end{equation}
In particular, we have for $\alpha = 1,\ldots, m$ that,
\begin{equation}
D_\alpha = \sum_{\beta=1}^m A_{\alpha \beta} x_\beta + \sum_{1\leq i \leq j \leq E} (B_\alpha)_{ij} \lambda_{ij} \hspace{3mm} \mathrm{with} \hspace{3mm}
A_{\alpha \beta} \in \mathbb{Z} \,,
\label{eq:relation_of_z_to_x}
\end{equation}
where $\beta = 1, \ldots, m$ labels the lexicographically-ordered
elements $(i,j)$ in eq.~(\ref{eq:independent_loop_containing_entries}),
and where the entries of $B_\alpha$ are integers.

The Baikov representation uses the inverse propagators and irreducible scalar
products as variables,
\begin{equation}
z_\alpha \equiv D_\alpha \hspace{4mm} \mathrm{where} \hspace{4mm} 1 \leq \alpha \leq m \,.
\label{eq:definition_of_z}
\end{equation}
The Jacobian associated with the change of variables from $(\ell_1^\mu, \ldots, \ell_L^\mu)$
in eq.~(\ref{eq:def_generic_Feynman_integral}) to $(z_1, \ldots, z_m)$
is an appropriate power of the determinant of $S$,
\begin{equation}
F \equiv \det S \,.
\label{eq:definition_of_Baikov_polynomial}
\end{equation}
The Baikov representation of the integral in eq.~(\ref{eq:def_generic_Feynman_integral})
takes the form
\footnote{
The consistency of the Baikov representation in eq.~(\ref{eq:Baikov_representation})
with that used in ref.~\cite{Larsen:2015ped} follows from the identity
$\det_{i,j=1,\ldots,L} \mu_{ij} = \frac{F}{U}$,
which in turn is a consequence of the Schur complement theorem in linear algebra.
Also note that ref.~\cite{Larsen:2015ped} uses the four-dimensional helicity
scheme and therefore requires that the external momenta span a vector space of
dimension at most four, i.e. $\mathrm{dim} \hspace{0.9mm} \mathrm{span} \{ p_1, \ldots, p_n\} \leq 4$.
As a result, the exponent of the Baikov polynomial there, with $L=2$, takes
the values $\frac{D-7}{2}$ and $\frac{D-6}{2}$ for $E \geq 4$ and $E=3$, respectively.
},
\begin{align}
I(N; \alpha; D) \hspace{0.5mm}&=\hspace{0.5mm} C_E^L(D) U^\frac{E-D+1}{2} \hspace{-1mm}
\int \frac{\d z_1 \cdots \d z_m}{z_1^{\alpha_1} \cdots z_m^{\alpha_m}} F^\frac{D-L-E-1}{2} N \,,
\label{eq:Baikov_representation}
\end{align}
where the first prefactor is given by,
\begin{equation}
C_E^L (D) \hspace{0.7mm} \equiv \hspace{0.7mm} \frac{\pi^{-L(L-1)/4 - L E/2}}
{\prod_{j=1}^L \Gamma \hspace{-0.5mm} \left(\frac{D-L-E+j}{2} \right)} \det A \,,
\label{eq:Baikov_prefactor}
\end{equation}
with $A$ defined in eq.~(\ref{eq:relation_of_z_to_x}),
and where $U$ denotes the determinant of the Gram matrix (\ref{eq:Gram_matrix_of_external_momenta})
of the independent external momenta,
\begin{equation}
U \equiv \det G \,.
\end{equation}
We remark that $U$ is equal to the square of the volume of the
parallelotope formed by the independent external momenta $\{ p_1, \ldots, p_E\}$
and therefore is nonvanishing for non-collinear external momenta.

\section{Differential equations in Baikov representation}\label{sec:diff_eqs_in_Baikov_representation}

In this section we show how differential equations of the form
in eq.~(\ref{eq:diff_eqs_schematic}) can be derived in Baikov representation.
We note that related work has appeared in refs.~\cite{Frellesvig:2017aai,Zeng:2017ipr,Harley:2017qut}.

Letting $\big(I_1, \ldots, I_M \big)$ denote a basis of integrals and
acting on the Baikov represention (\ref{eq:Baikov_representation}) by
a derivative with respect to an arbitrary external invariant $\chi$, we find
\footnote{We remark that the domain of integration in
eq.~(\ref{eq:Baikov_representation}) depends on $\chi$.
Hence one expects a term corresponding to the domain dependence
on $\chi$ to appear in the derivative. However, this term vanishes
since the integrand is zero on the boundary of the domain.
We refer to Sec.~2 of ref.~\cite{Bosma:2017ens} for a discussion
of this point.},
\begin{align}
\frac{\partial}{\partial \chi} I_j (N_j;\alpha;D)
&= \frac{E-D+1}{2U} \frac{\partial U}{\partial \chi} I_j (N_j;\alpha;D) \nonumber \\[1.5mm]
&\hspace{-10mm} +
\frac{D{-}L{-}E{-}1}{2} I_j \Big( \frac{1}{F} \frac{\partial F}{\partial \chi} N_j; \alpha; D \Big) \,.
\label{eq:diff_eqs_in_Baikov_rep_1}
\end{align}
We observe that the $\frac{1}{F}$ factor in the
second line effectively modifies the integration measure in
eq.~(\ref{eq:Baikov_representation}), shifting
the space-time dimension from $D$ to $D-2$.

This motivates us to ask whether there exist
polynomials $(a_1, \ldots, a_m, b)$ such that the following relation holds,
\begin{equation}
\frac{\partial F}{\partial \chi} = \sum_{i=1}^m a_i \frac{\partial F}{\partial z_i} + bF \,,
\label{eq:Baikov_poly_ideal_membership_explicit_1}
\end{equation}
since, as we will show shortly, this implies that
dimension shifts can be avoided in eq.~(\ref{eq:diff_eqs_in_Baikov_rep_1}).
In Appendix~\ref{sec:diff_eqs_via_dimension_shifts} we work out the
form of the differential equations (\ref{eq:diff_eqs_in_Baikov_rep_1})
when $\frac{\partial F}{\partial \chi}$ is left unchanged
by applying dimension-shift identities---i.e., relations between integrals
in $D$ and $D-2$ dimensions.

The relation~(\ref{eq:Baikov_poly_ideal_membership_explicit_1})
can equivalently be stated as the ideal membership property,
\begin{equation}
\frac{\partial F}{\partial \chi} \in \left\langle \frac{\partial F}{\partial z_1}, \ldots,
\frac{\partial F}{\partial z_m}, F \right\rangle \,,
\label{eq:Baikov_poly_ideal_membership_1}
\end{equation}
where the ideal is understood to be embedded in the ring
$\mathbb{C}[z_1, \ldots, z_m, \lambda_{ij}]$
of polynomials in the Baikov variables $(z_1, \ldots, z_m)$ and
external invariants $(\lambda_{i,j})_{i,j=1,\ldots,E}$ over $\mathbb{C}$.

In the following we will refer to eq.~(\ref{eq:Baikov_poly_ideal_membership_1})
as the \emph{fundamental} ideal membership.
In Sec.~\ref{sec:proof_of_fundamental_ideal_membership} we give a proof
of eq.~(\ref{eq:Baikov_poly_ideal_membership_1}). The proof is constructive and gives an
explicit construction of the cofactors $(a_1, \ldots, a_m, b)$
in eq.~(\ref{eq:Baikov_poly_ideal_membership_explicit_1}).

Let us now turn to the consequence of the ideal membership
for the differential equations.
Inserting eq.~(\ref{eq:Baikov_poly_ideal_membership_explicit_1}) into
eq.~(\ref{eq:diff_eqs_in_Baikov_rep_1}) and using
elementary integration by parts in $z_i$,
\begin{equation}
\frac{a_i N_j}{z_1^{\alpha_1} \cdots z_m^{\alpha_m}}
\frac{\partial}{\partial z_i} F^\frac{D-L-E-1}{2}
= - \frac{F^\frac{D-L-E-1}{2}}{z_1^{\alpha_1} \cdots z_m^{\alpha_m}}
z_i^{\alpha_i} \frac{\partial}{\partial z_i} \frac{a_i N_j}{z_i^{\alpha_i}} \,,
\end{equation}
we find,
\begin{align}
\frac{\partial}{\partial \chi} I_j (N_j;\alpha;D)
&= \frac{E-D+1}{2 U} \frac{\partial U}{\partial \chi} I_j (N_j;\alpha;D) \nonumber \\
&\hspace{20mm} + I_j \left( Q_j  ; \alpha; D \right) \,,
\label{eq:differential_equations}
\end{align}
where the insertion in the second term is given by,
\begin{equation}
Q_j = - \sum_{i=1}^m z_i^{\alpha_i} \frac{\partial}{\partial z_i} \hspace{-0.7mm}
\left(\frac{a_i N_j}{z_i^{\alpha_i}} \right) + \frac{D{-}L{-}E{-}1}{2} b N_j \,.
\label{eq:numerator_insertion}
\end{equation}
We observe that all integrals in the differential
equations (\ref{eq:differential_equations}) are $D$-dimensional integrals,
so that no dimension-shift identities are required.

By use of integration-by-parts reductions we can express
the integrals in the second line of eq.~(\ref{eq:differential_equations})
as linear combinations of the basis integrals $I_j (N_j; \alpha; D)$,
\begin{equation}
I_j (Q_j; \alpha; D) = \sum_{k=1}^M R_{jk} (\lambda, D) I_k(N_k; \alpha; D) \,.
\end{equation}
Thus, we arrive at the system of differential equations,
\begin{equation}
\frac{\partial}{\partial \chi} I_j (N_j;\alpha;D)
= \sum_{k=1}^M A_{jk} (\lambda, D) I_k (N_k;\alpha;D) \,,
\end{equation}
where the coefficient matrix is given by,
\begin{equation}
A_{jk} (\lambda, D) = \frac{E-D+1}{2 U} \frac{\partial U}{\partial \chi} \delta_{jk}
+ R_{jk} (\lambda, D) \,.
\end{equation}

\section{Proof of fundamental ideal membership}\label{sec:proof_of_fundamental_ideal_membership}

In this section we give a proof of the fundamental ideal membership
(\ref{eq:Baikov_poly_ideal_membership_1})
\footnote{We thank Roman N.~Lee for his idea of proving this ideal membership conjecture via
Laplace expansion of the determinant of symmetric matrices.}.
As explained in Sec.~\ref{sec:diff_eqs_in_Baikov_representation}, the latter implies
that differential equations of the form in eq.~(\ref{eq:diff_eqs_schematic})
can always be derived in Baikov representation without involving dimension-shift
identities---i.e., relations between integrals in $D$ and $D-2$ dimensions
of the form in eq.~(\ref{eq:D-dim_shift_identities}).

We emphasize that the proof is constructive and gives an
explicit construction of the cofactors $(a_1, \ldots, a_m, b)$
in eq.~(\ref{eq:Baikov_poly_ideal_membership_explicit_1}).

\begin{theorem}
Taking the Baikov polynomial in eq.~\textup{(\ref{eq:definition_of_Baikov_polynomial})} to
depend on the Baikov variables in eq.~\textup{(\ref{eq:definition_of_z})} and the $\lambda_{ij}$
in eq.~\textup{(\ref{eq:definition_of_lambda})} and letting $\chi$ denote any of the $\lambda_{ij}$,
there exist polynomials $(a_i, b)$ such that the following relation holds,
\begin{equation}
\frac{\partial F}{\partial \chi} = \sum_{i=1}^m a_i \frac{\partial F}{\partial z_i} + bF \,.
\end{equation}
\end{theorem}
\begin{proof}
Given a generic matrix $R = (r_{ij})_{i,j=1,\ldots, n}$ whose entries
are all independent, the Laplace expansion of the determinant
along the $i$th row can be expressed in the form,
\begin{equation}
\left[ \sum_{k=1}^n r_{jk} \frac{\partial (\det R)}{\partial r_{ik}} \right] - \delta_{ij} \det R = 0 \,,
\hspace{4mm} 1 \leq i,j \leq n \,.
\label{eq:Laplace_expansion_of_generic_matrix}
\end{equation}
The identities with $i\neq j$ arise by replacing the $i$th row of $R$
by the $j$th row, $r_{ik} \to r_{jk}$, which clearly produces a
matrix with a vanishing determinant.

For a symmetric matrix $S = (s_{ij})_{i,j=1,\ldots, n}$ the entries satisfy
$s_{ij}=s_{ji}$ and thus are not independent, but rather linearly constrained.
In this case, the Laplace expansion produces the identities,
\begin{equation}
\left[\sum_{k=1}^n (1{+}\delta_{ik}) s_{jk} \frac{\partial (\det S)}{\partial s_{ik}}\right] - 2\delta_{ij} \det S = 0 \,,
\hspace{2mm} 1 \leq i,j \leq n \,.
\label{eq:Laplace_expansion_of_symmetric_matrix}
\end{equation}
In taking the derivatives we must bear in mind that the entries
are not independent. We therefore imagine making the replacement
$s_{ji} \to s_{ij}$ with $i \leq j$ in $S$ prior to taking
derivatives, and in turn interpret $\frac{\partial (\det S)}{\partial s_{ik}}$
with $i>k$ as $\frac{\partial (\det S)}{\partial s_{ki}}$.

For the Gram matrix $S$ in eq.~(\ref{eq:extended_Gram_matrix}),
the identity (\ref{eq:Laplace_expansion_of_symmetric_matrix}) implies
for $1 \leq i,j \leq E$ that,
\begin{equation}
\sum_{k=1}^E (1{+}\delta_{ik}) \lambda_{jk}
\left(\hspace{-0.5mm}\frac{\partial F}{\partial \lambda_{ik}}\hspace{-0.5mm}\right)_{\hspace{-1mm} x}  = 2\delta_{ij} F - \sum_{q=E+1}^{E+L} (1{+}\delta_{iq}) x_{jq}
\frac{\partial F}{\partial x_{iq}} \,.
\label{eq:external_derivatives_of_F}
\end{equation}
The notation $\left(\hspace{-0.5mm}\frac{\partial F}{\partial \lambda_{ik}}\hspace{-0.5mm}\right)_{\hspace{-1mm} x}$
is used to emphasize that $F$ is here treated as a function of $x$ and $\lambda$,
and the $x$-variables are treated as constants when evaluating the derivative.
On the other hand, we may use the relations in
eqs.~(\ref{eq:relation_of_z_to_x})--(\ref{eq:definition_of_z})
to express $F$ as a function of $z$ and $\lambda$. The corresponding derivative
obtained while treating the $z$-variables as constants is related to
the former derivative through the chain rule,
\begin{equation}
\left(\hspace{-0.5mm}\frac{\partial F}{\partial \lambda_{ik}}\hspace{-0.5mm}\right)_{\hspace{-1mm} x}
= \left(\hspace{-0.5mm}\frac{\partial F}{\partial \lambda_{ik}}\hspace{-0.5mm}\right)_{\hspace{-1mm} z}
+ \sum_{\alpha=1}^m \frac{\partial z_\alpha}{\partial \lambda_{ik}}
\frac{\partial F}{\partial z_\alpha} \hspace{3mm} \mathrm{for} \hspace{3mm}
1 \leq i, k \leq E \,.
\label{eq:relation_between_F_derivatives_with_constant_x_and_z}
\end{equation}
Now, for a fixed $1 \leq i \leq E$, consider the left-hand side of
eq.~(\ref{eq:external_derivatives_of_F}) as the product of a matrix
$(G_i)_{jk}$ and a vector
$\left(\hspace{-0.5mm}\frac{\partial F}{\partial \lambda_{ik}}\hspace{-0.5mm}\right)_{\hspace{-1mm} x}$.
That is, by defining for a fixed $1 \leq i \leq E$ the matrix,
\begin{equation}
(G_i)_{jk} = (1 + \delta_{ik}) \lambda_{jk} \hspace{3mm} \mathrm{for} \hspace{3mm}
1 \leq j, k \leq E \,,
\label{eq:modified_Gram_matrix_of_external_moment}
\end{equation}
we can recast eq.~(\ref{eq:external_derivatives_of_F}) in the form,
\begin{equation}
\sum_{k=1}^E (G_i)_{jk}
\left(\hspace{-0.5mm}\frac{\partial F}{\partial \lambda_{ik}}\hspace{-0.5mm}\right)_{\hspace{-1mm} x}  = 2\delta_{ij} F - \sum_{q=E+1}^{E+L} (1{+}\delta_{iq}) x_{jq}
\frac{\partial F}{\partial x_{iq}} \,.
\label{eq:external_derivatives_of_F_recast}
\end{equation}
From eq.~(\ref{eq:modified_Gram_matrix_of_external_moment}) we observe
that $G_i$ is simply the Gram matrix $G$ in eq.~(\ref{eq:Gram_matrix_of_external_momenta})
with the $i$th column multiplied by a factor of 2. This implies that
$G_i$ is invertible, since $\det G_i = 2\det G = 2 U \neq 0$
for non-collinear external momenta.

Multiplying eq.~(\ref{eq:external_derivatives_of_F_recast}) by $G_i^{-1}$ from the left
and using eq.~(\ref{eq:relation_between_F_derivatives_with_constant_x_and_z})
and the chain rule,
\begin{equation}
\frac{\partial F}{\partial x_{iq}} = \sum_{\alpha=1}^m \frac{\partial z_\alpha}{\partial x_{iq}}
\frac{\partial F}{\partial z_\alpha} \hspace{4mm} \mathrm{for} \hspace{2mm}
\left\{ \hspace{-0.7mm}\begin{array}{l}
1 \leq i \leq E \,, \\[1.5mm]
E{+}1 \leq q \leq E{+}L \,,
\end{array} \right.
\end{equation}
we find that,
\begin{equation}
\left(\hspace{-0.5mm}\frac{\partial F}{\partial \lambda_{ik}}\hspace{-0.5mm}\right)_{\hspace{-1mm} z}
= \left[ \sum_{\alpha=1}^m a_{ik, \alpha} \frac{\partial F}{\partial z_\alpha} \right] + b_{ik} F \,,
\end{equation}
where the cofactors are given by,
\begin{align}
a_{ik, \alpha} &= -\frac{\partial z_\alpha}{\partial \lambda_{ik}}
- \sum_{j=1}^E \sum_{q=E+1}^{E+L} \hspace{-1mm} (1{+}\delta_{iq})
\frac{\partial z_\alpha}{\partial x_{iq}} \big(G_i^{-1}\big)_{kj} x_{jq} \,, \\
b_{ik}         &= 2 \big(G_i^{-1}\big)_{ki} \,.
\end{align}
From the relations in eqs.~(\ref{eq:relation_of_z_to_x})--(\ref{eq:definition_of_z})
it follows that the derivatives $\frac{\partial z_\alpha}{\partial \lambda_{ik}}$
and $\frac{\partial z_\alpha}{\partial x_{iq}}$ are integers. Furthermore, we may
use the relations to express the $x$-variables as a linear combination of the
$z$-variables.
This proves the theorem and moreover shows that the cofactors
$(a_1, \ldots, a_m, b)$ in eq.~(\ref{eq:Baikov_poly_ideal_membership_explicit_1})
can be taken to be at most linear polynomials in the Baikov variables
$z_\alpha$.

In the above proof, we assumed that all the $\lambda_{ij}$ are independent,
as is the case for generic kinematics. For massless kinematics, the
chain rules between independent Mandelstam invariants and $\lambda_{ij}$ imply that the
conclusions about the fundamental ideal membership and the degrees of the
cofactors continue to hold.
\end{proof}

\section{Differential equations without squared propagators}\label{sec:diff_eqs_without_squared_propagators}

In this section we investigate to what extent the IBP reduction formalisms
of refs.~\cite{Gluza:2010ws,Ita:2015tya,Larsen:2015ped,Bern:2017gdk}
are compatible with differential equations of the type in
eq.~(\ref{eq:diff_eqs_schematic}).
As we will see, for a large class of multi-loop diagrams,
it is possible to avoid having the derivatives on the left-hand side
of eq.~(\ref{eq:diff_eqs_schematic}) produce integrals
with squared propagators. We note that related work has appeared in refs.~\cite{Zeng:2017ipr,Bern:2017gdk}.

Returning to the derivation in Sec.~\ref{sec:diff_eqs_in_Baikov_representation},
from eq.~(\ref{eq:numerator_insertion}) we observe
that terms with positive $\alpha_i$ will produce squared
propagators for a generic polynomial $a_i$.

However, provided it
is possible to choose the polynomials $a_i$ such that,
\begin{equation}
a_i = z_i b_i \hspace{5mm} \mathrm{for} \hspace{5mm} i=1,\ldots,k \,,
\label{eq:cofactor_proportionality}
\end{equation}
where $b_i$ denote polynomials, the insertion (\ref{eq:numerator_insertion})
takes the following form,
\begin{align}
Q_j = & - \sum_{i=1}^k \left[ z_i \frac{\partial (b_i N_j)}{\partial z_i} + (1{-}\alpha_i) b_i N_j \right] \nonumber\\
&-\sum_{i=k+1}^m \left[ \frac{\partial (a_i N_j)}{\partial z_i}
- \alpha_i \frac{a_i N_j}{z_i} \right] + \frac{D{-}L{-}E{-}1}{2} b N_j \,.
\label{eq:numerator_of_differentiated_integrand}
\end{align}
We observe that the decomposition of $\frac{\partial F}{\partial \chi}$ in
eq.~(\ref{eq:Baikov_poly_ideal_membership_explicit_1}) with this choice of
cofactors $a_i$ leads to a set of differential equations
eq.~(\ref{eq:differential_equations}) without squared propagators.

In the following we will show that the \emph{enhanced} ideal membership condition,
equivalent to eq.~(\ref{eq:Baikov_poly_ideal_membership_explicit_1})
combined with eq.~(\ref{eq:cofactor_proportionality}),
\begin{equation}
\frac{\partial F}{\partial \chi} \in \left\langle z_1 \frac{\partial F}{\partial z_1}, \ldots,
z_k \frac{\partial F}{\partial z_k}, \frac{\partial F}{\partial z_{k+1}}, \ldots,
\frac{\partial F}{\partial z_m},F \right\rangle \,,
\label{eq:Baikov_poly_ideal_membership_2}
\end{equation}
turns out to hold for a large class of multi-loop diagrams.

We note that ideal membership can be determined by computing
a Gr{\"o}bner basis $\mathcal{G}$ of the ideal
on the right-hand side of eq.~(\ref{eq:Baikov_poly_ideal_membership_2})
and computing the remainder $r$ of $\frac{\partial F}{\partial \chi}$
after polynomial division with respect to $\mathcal{G}$.
Namely, eq.~(\ref{eq:Baikov_poly_ideal_membership_2}) holds
if and only if $r=0$.

Alternatively, one can solve explicitly for the cofactors
$(b_1, \ldots, b_k, a_{k+1}, \ldots, a_m, b)$
by starting with Ans{\"a}tze which are linear in $(z_1, \ldots, z_m)$
and iteratively allowing for solutions of higher degree. This turns
out to be an efficient approach in practice, as cofactors are
found to be of low degrees and hence lead to linear systems of manageable sizes.

\section{Examples}\label{sec:examples}

In this section we work out several examples of the formalism
developed in Secs.~\ref{sec:diff_eqs_in_Baikov_representation}--\ref{sec:diff_eqs_without_squared_propagators}
and present explicit results for the cofactors associated with
the enhanced ideal membership (\ref{eq:Baikov_poly_ideal_membership_2})
along with the resulting differential equations.

We display a selection of multi-loop diagrams
for which the enhanced ideal membership
(\ref{eq:Baikov_poly_ideal_membership_2}) has been verified to hold,
and also show a counterexample.

\subsection{Massless planar double-box}

As a first simple example we consider the fully massless planar
double-box diagram shown in fig.~\ref{fig:massless_planar_DB_z_variables}.

\begin{figure}[!h]
\begin{center}
\includegraphics[angle=0, width=0.28\textwidth]{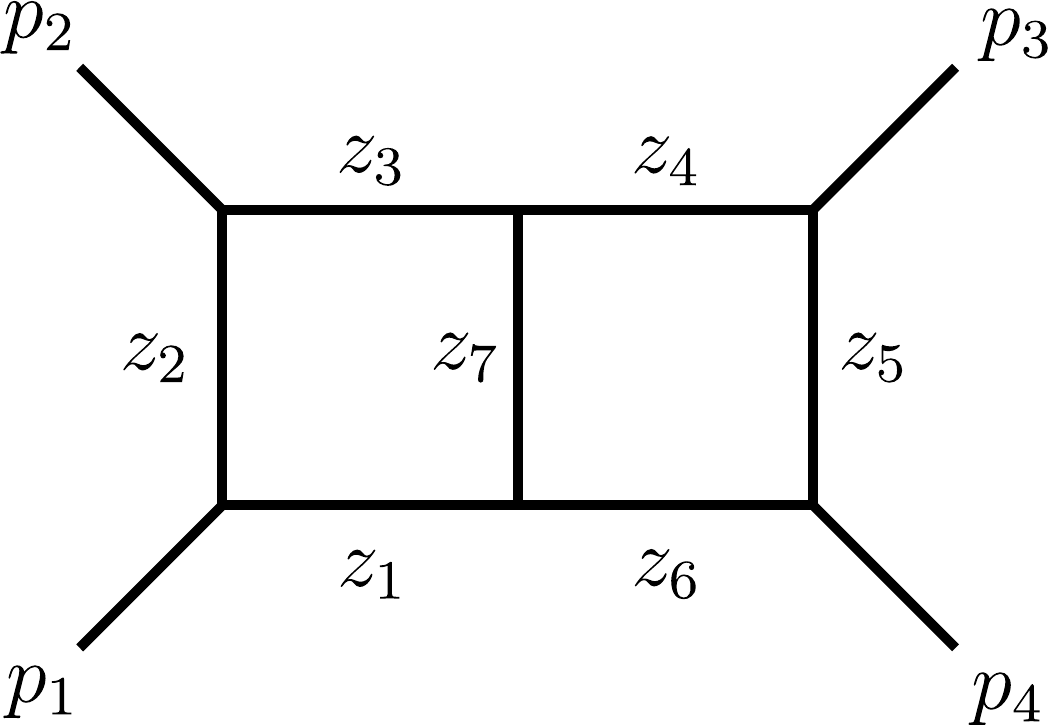}
{\vskip -0mm}
\caption{The fully massless planar double-box diagram.
All external momenta are taken to be outgoing.}
\label{fig:massless_planar_DB_z_variables}
\end{center}
\end{figure}

\noindent In agreement with eq.~(\ref{eq:definition_of_z}),
we define the $z$-variables as follows, setting
$P_{12} \equiv p_1 + p_2$,
\begin{equation}
\begin{alignedat}{3}
z_1&=\ell_1^2\,,           \hspace{3mm}  && z_2 = (\ell_1 - p_1)^2\,,  \hspace{3mm} && z_3 = (\ell_1 - P_{12})^2\,, \\
z_4&=(\ell_2+P_{12})^2\,,  \hspace{3mm}  && z_5 = (\ell_2 - p_4)^2\,,  \hspace{3mm} && z_6 = \ell_2^2\,, \\
z_7&=(\ell_1+\ell_2)^2\,,  \hspace{3mm}  && z_8 = (\ell_1 + p_4)^2\,,  \hspace{3mm} && z_9 = (\ell_2 + p_1)^2\,.
\end{alignedat}
\end{equation}
In terms of these variables and the Mandelstam invariants $s=(p_1 + p_2)^2$ and $t=(p_2 + p_3)^2$,
the Gram determinant in eq.~(\ref{eq:definition_of_Baikov_polynomial}) takes the following form,
\begin{equation}
F=\begin{vmatrix}
 0 & \frac{s}{2} & -\frac{s+t}{2} & \frac{z_1-z_2}{2} & \frac{z_9-z_6}{2} \\[2mm]
 \frac{s}{2} & 0 & \frac{t}{2} & \frac{z_2-z_3+s}{2} & \frac{z_4-z_9-s}{2} \\[2mm]
 -\frac{s+t}{2} & \frac{t}{2} & 0 & \frac{z_3-z_8-s}{2} & \frac{z_5 - z_4 + s}{2} \\[2mm]
 \frac{z_1-z_2}{2} & \frac{z_2 - z_3 +s}{2} & \frac{z_3-z_8-s}{2} & z_1 & \frac{z_7 - z_1 - z_6}{2} \\[2mm]
 \frac{z_9-z_6}{2} & \frac{z_4 - z_9 -s}{2} & \frac{z_5 - z_4 +s}{2} & \frac{z_7 - z_1 - z_6}{2} & z_6
\end{vmatrix} \,.
\label{eq:extended_Gram_det_planar_DB}
\end{equation}

We are interested in setting up differential equations for a basis
of the vector space spanned by the diagram in fig.~\ref{fig:massless_planar_DB_z_variables}
and its subdiagrams. An integral basis $\boldsymbol{\mathcal{I}}$ can be obtained
with \textsc{Azurite} \cite{Georgoudis:2016wff} which produces,
setting $I_\alpha \equiv I(1; \alpha; D)$ for brevity,
\begin{alignat}{3}
\boldsymbol{\mathcal{I}} = \Big( s^{-1 + 2\epsilon} I_{(0,1,0,0,1,0,1,0,0)},& \hspace{3mm}
&& s^{-1 + 2\epsilon} I_{(1,0,0,1,0,0,1,0,0)}, \nonumber \\
s^{2\epsilon} I_{(1,0,1,0,1,0,1,0,0)},& \hspace{3mm}
&& s^{2\epsilon} I_{(1,0,1,1,0,1,0,0,0)}, \nonumber \\
s^{1 + 2\epsilon} I_{(1,1,0,1,1,0,1,0,0)},& \hspace{3mm}
&& s^{1 + 2\epsilon} I_{(1,1,1,0,1,0,1,0,0)}, \nonumber \\
s^{3 + 2\epsilon} I_{(1,1,1,1,1,1,1,0,0)},& \hspace{3mm}
&& s^{2 + 2\epsilon} I_{(1,1,1,1,1,1,1,-1,0)} \Big) \,.
\label{eq:planar_DB_integral_basis}
\end{alignat}
Here we rescaled the basis integrals $I_\alpha$ by appropriate
factors of $s^{|\alpha| - 4 + 2\epsilon}$ to render them dimensionless.
The basis integrals can thus only depend on the external kinematics
through dimensionless ratios, of which only $\chi \equiv t/s$
is available.

We are therefore interested in differential equations for
the basis integrals in eq.~(\ref{eq:planar_DB_integral_basis})
taken with respect to $\chi$. To this end, we set $t = \chi s$
in eq.~(\ref{eq:extended_Gram_det_planar_DB}).

Now, in order to avoid introducing integrals with squared
propagators in intermediate stages, we ask whether
the enhanced ideal membership in eq.~(\ref{eq:Baikov_poly_ideal_membership_2}) holds.
We find that in the case at hand,
a slightly stronger ideal membership holds,
\begin{equation}
\frac{\partial F}{\partial \chi} =
\sum_{i=1}^9 b_i z_i \frac{\partial F}{\partial z_i} + bF \,.
\label{eq:decomposed_derivative_of_Baikov_pol}
\end{equation}
The cofactors are, setting $\mathbf{b} = (b_1, \ldots, b_9)$,
\begin{align}
\mathbf{b} &= \left({\textstyle\frac{z_3-z_8}{\chi (\chi +1) s}},
   {\textstyle\frac{z_3 - z_8 - \chi s -s}{\chi (\chi + 1) s}},
   {\textstyle\frac{z_3 - z_8 - s}{\chi (\chi+1) s}},
   {\textstyle\frac{z_4 - z_5 - s}{\chi (\chi +1) s}},
   {\textstyle\frac{z_4 - z_5 - s}{\chi (\chi +1) s}}, \right. \nonumber \\
&\hspace{7mm} \left.
   {\textstyle\frac{z_4-z_5}{\chi (\chi+1) s}},
   {\textstyle\frac{z_3 + z_4 - z_5 - z_8 - s}{\chi (\chi +1) s}},
   {\textstyle\frac{z_3 - z_8 - s}{\chi (\chi +1) s}},
   {\textstyle\frac{z_4 - z_5 - \chi s - s}{\chi (\chi +1) s}} \right) \nonumber \\[2mm]
b &= -{\textstyle\frac{2z_3 + 2z_4 - 2z_5 - 2z_8 - 2\chi s - 3s}{\chi (\chi +1) s}} \,.
\label{eq:massless_planar_four-point_cofactors}
\end{align}
These cofactors were found by writing Ans{\"a}tze for $(b_i,b)$
which are linear in $(z_1, \ldots, z_9)$ and solving the resulting
linear system. Having obtained expressions for the cofactors, we insert them into
eq.~(\ref{eq:numerator_of_differentiated_integrand}) to obtain the explicit
right-hand side of eq.~(\ref{eq:differential_equations}). After applying
integration-by-parts reductions to the resulting right-hand side, we find
a system of differential equations of the desired form,
\begin{equation}
\frac{\partial}{\partial \chi} \boldsymbol{\mathcal{I}} (\chi, \epsilon)
= A(\chi, \epsilon) \boldsymbol{\mathcal{I}} (\chi, \epsilon) \,,
\label{eq:diff_eqs_single_ratio}
\end{equation}
where we have set $\epsilon = \frac{4-D}{2}$, as we are interested
in computing the integrals in a Laurent expansion around
four space-time dimensions.
The resulting coefficient matrix $A(\chi, \epsilon)$ is not particularly
illuminating, and rather than presenting its explicit form, we take
one further step \cite{Henn:2013pwa} and rotate
to a basis in which the coefficient matrix becomes proportional
to $\epsilon$.

To this end, we note that under a change of integral basis,
\begin{equation}
\boldsymbol{\mathcal{J}}(\chi, \epsilon) = U(\chi, \epsilon) \boldsymbol{\mathcal{I}}(\chi, \epsilon) \,,
\label{eq:change-of-basis_matrix}
\end{equation}
the system of differential equations (\ref{eq:diff_eqs_single_ratio}) becomes,
\begin{equation}
\frac{\partial}{\partial \chi} \boldsymbol{\mathcal{J}} (\chi, \epsilon)
= \widehat{A}(\chi, \epsilon) \boldsymbol{\mathcal{J}} (\chi, \epsilon) \,,
\end{equation}
where the transformed coefficient matrix is,
\begin{equation}
\widehat{A} = U A U^{-1} + \frac{\partial U}{\partial \chi} U^{-1} \,.
\end{equation}
We can find a change-of-basis matrix $U$ with the desired property
by using \texttt{Fuchsia} \cite{Gituliar:2017vzm}.
Providing as input the coefficient matrix $A(\chi, \epsilon)$ computed
in eq.~(\ref{eq:diff_eqs_single_ratio}), it finds the following
transformation matrix,
\begin{align}
U &= \mathrm{diag} \Big( \textstyle{\frac{(1 - 2 \epsilon) (1 - 3 \epsilon) (2 - 3 \epsilon)}{120 \epsilon^3 \chi}},
\textstyle{\frac{(1 - 2 \epsilon) (1 - 3 \epsilon) (-2 + 3 \epsilon)}{120 \epsilon^3}}, \nonumber \\
& \hspace{7mm} \textstyle{\frac{(1 - 2 \epsilon) (1 - 3 \epsilon)}{24 \epsilon^2}},
\textstyle{\frac{(1 - 2 \epsilon)^2}{18 \epsilon^2}},
-\textstyle{\frac{\chi+1}{2}},
\textstyle{\frac{-1 + 2 \epsilon}{6 \epsilon}},
-\textstyle{\frac{\chi}{2}},
\textstyle{\frac{1}{2}} \Big) \,.
\end{align}
The transformed coefficient matrix is proportional to $\epsilon$,
as desired, and takes the following form,
\begin{equation}
\widehat{A} = \epsilon \left( \frac{a_0}{\chi} + \frac{a_{-1}}{\chi + 1} \right) \,,
\label{eq:canonical_coeff_matrix_planar_DB}
\end{equation}
where $a_0$ and $a_{-1}$ are matrices with integer entries,
\begin{equation}
a_0 = \begin{pmatrix}
 -2 \hspace{-0.5mm} & \hspace{-0.5mm} 0 \hspace{-0.5mm} & \hspace{-0.5mm} 0 \hspace{-0.5mm} & \hspace{-0.5mm} 0 \hspace{-0.5mm} & \hspace{-0.5mm} 0 \hspace{-0.5mm} & \hspace{-0.5mm} 0 \hspace{-0.5mm} & \hspace{-0.5mm} 0 \hspace{-0.5mm} & \hspace{-0.5mm} 0 \\
 0 \hspace{-0.5mm} & \hspace{-0.5mm} 0 \hspace{-0.5mm} & \hspace{-0.5mm} 0 \hspace{-0.5mm} & \hspace{-0.5mm} 0 \hspace{-0.5mm} & \hspace{-0.5mm} 0 \hspace{-0.5mm} & \hspace{-0.5mm} 0 \hspace{-0.5mm} & \hspace{-0.5mm} 0 \hspace{-0.5mm} & \hspace{-0.5mm} 0 \\
 0 \hspace{-0.5mm} & \hspace{-0.5mm} 0 \hspace{-0.5mm} & \hspace{-0.5mm} 0 \hspace{-0.5mm} & \hspace{-0.5mm} 0 \hspace{-0.5mm} & \hspace{-0.5mm} 0 \hspace{-0.5mm} & \hspace{-0.5mm} 0 \hspace{-0.5mm} & \hspace{-0.5mm} 0 \hspace{-0.5mm} & \hspace{-0.5mm} 0 \\
 0 \hspace{-0.5mm} & \hspace{-0.5mm} 0 \hspace{-0.5mm} & \hspace{-0.5mm} 0 \hspace{-0.5mm} & \hspace{-0.5mm} 0 \hspace{-0.5mm} & \hspace{-0.5mm} 0 \hspace{-0.5mm} & \hspace{-0.5mm} 0 \hspace{-0.5mm} & \hspace{-0.5mm} 0 \hspace{-0.5mm} & \hspace{-0.5mm} 0 \\
 -60 \hspace{-0.5mm} & \hspace{-0.5mm} -60 \hspace{-0.5mm} & \hspace{-0.5mm} 0 \hspace{-0.5mm} & \hspace{-0.5mm} 0 \hspace{-0.5mm} & \hspace{-0.5mm} -2 \hspace{-0.5mm} & \hspace{-0.5mm} 0 \hspace{-0.5mm} & \hspace{-0.5mm} 0 \hspace{-0.5mm} & \hspace{-0.5mm} 0 \\
 20 \hspace{-0.5mm} & \hspace{-0.5mm} 0 \hspace{-0.5mm} & \hspace{-0.5mm} -4 \hspace{-0.5mm} & \hspace{-0.5mm} 0 \hspace{-0.5mm} & \hspace{-0.5mm} 0 \hspace{-0.5mm} & \hspace{-0.5mm} -2 \hspace{-0.5mm} & \hspace{-0.5mm} 0 \hspace{-0.5mm} & \hspace{-0.5mm} 0 \\
 -360 \hspace{-0.5mm} & \hspace{-0.5mm} 360 \hspace{-0.5mm} & \hspace{-0.5mm} 72 \hspace{-0.5mm} & \hspace{-0.5mm} 0 \hspace{-0.5mm} & \hspace{-0.5mm} 12 \hspace{-0.5mm} & \hspace{-0.5mm} 36 \hspace{-0.5mm} & \hspace{-0.5mm} -2 \hspace{-0.5mm} & \hspace{-0.5mm} 0 \\
 540 \hspace{-0.5mm} & \hspace{-0.5mm} -360 \hspace{-0.5mm} & \hspace{-0.5mm} -90 \hspace{-0.5mm} & \hspace{-0.5mm} -9 \hspace{-0.5mm} & \hspace{-0.5mm} -18 \hspace{-0.5mm} & \hspace{-0.5mm} -36 \hspace{-0.5mm} & \hspace{-0.5mm} 1 \hspace{-0.5mm} & \hspace{-0.5mm} 1 \\
\end{pmatrix} \,,
\end{equation}
and
\begin{equation}
a_{-1} = \begin{pmatrix}
 0 \hspace{-0.5mm} & \hspace{-0.5mm} 0 \hspace{-0.5mm} & \hspace{-0.5mm} 0 \hspace{-0.5mm} & \hspace{-0.5mm} 0 \hspace{-0.5mm} & \hspace{-0.5mm} 0 \hspace{-0.5mm} & \hspace{-0.5mm} 0 \hspace{-0.5mm} & \hspace{-0.5mm} 0 \hspace{-0.5mm} & \hspace{-0.5mm} 0 \\
 0 \hspace{-0.5mm} & \hspace{-0.5mm} 0 \hspace{-0.5mm} & \hspace{-0.5mm} 0 \hspace{-0.5mm} & \hspace{-0.5mm} 0 \hspace{-0.5mm} & \hspace{-0.5mm} 0 \hspace{-0.5mm} & \hspace{-0.5mm} 0 \hspace{-0.5mm} & \hspace{-0.5mm} 0 \hspace{-0.5mm} & \hspace{-0.5mm} 0 \\
 0 \hspace{-0.5mm} & \hspace{-0.5mm} 0 \hspace{-0.5mm} & \hspace{-0.5mm} 0 \hspace{-0.5mm} & \hspace{-0.5mm} 0 \hspace{-0.5mm} & \hspace{-0.5mm} 0 \hspace{-0.5mm} & \hspace{-0.5mm} 0 \hspace{-0.5mm} & \hspace{-0.5mm} 0 \hspace{-0.5mm} & \hspace{-0.5mm} 0 \\
 0 \hspace{-0.5mm} & \hspace{-0.5mm} 0 \hspace{-0.5mm} & \hspace{-0.5mm} 0 \hspace{-0.5mm} & \hspace{-0.5mm} 0 \hspace{-0.5mm} & \hspace{-0.5mm} 0 \hspace{-0.5mm} & \hspace{-0.5mm} 0 \hspace{-0.5mm} & \hspace{-0.5mm} 0 \hspace{-0.5mm} & \hspace{-0.5mm} 0 \\
 0 \hspace{-0.5mm} & \hspace{-0.5mm} 0 \hspace{-0.5mm} & \hspace{-0.5mm} 0 \hspace{-0.5mm} & \hspace{-0.5mm} 0 \hspace{-0.5mm} & \hspace{-0.5mm} 2 \hspace{-0.5mm} & \hspace{-0.5mm} 0 \hspace{-0.5mm} & \hspace{-0.5mm} 0 \hspace{-0.5mm} & \hspace{-0.5mm} 0 \\
 -20 \hspace{-0.5mm} & \hspace{-0.5mm} 0 \hspace{-0.5mm} & \hspace{-0.5mm} 4 \hspace{-0.5mm} & \hspace{-0.5mm} 0 \hspace{-0.5mm} & \hspace{-0.5mm} 0 \hspace{-0.5mm} & \hspace{-0.5mm} 1 \hspace{-0.5mm} & \hspace{-0.5mm} 0 \hspace{-0.5mm} & \hspace{-0.5mm} 0 \\
 360 \hspace{-0.5mm} & \hspace{-0.5mm} -720 \hspace{-0.5mm} & \hspace{-0.5mm} -36 \hspace{-0.5mm} & \hspace{-0.5mm} 18 \hspace{-0.5mm} & \hspace{-0.5mm} -12 \hspace{-0.5mm} & \hspace{-0.5mm} -36 \hspace{-0.5mm} & \hspace{-0.5mm} 2 \hspace{-0.5mm} & \hspace{-0.5mm} 2 \\
 -540 \hspace{-0.5mm} & \hspace{-0.5mm} 360 \hspace{-0.5mm} & \hspace{-0.5mm} 90 \hspace{-0.5mm} & \hspace{-0.5mm} -9 \hspace{-0.5mm} & \hspace{-0.5mm} 18 \hspace{-0.5mm} & \hspace{-0.5mm} 36 \hspace{-0.5mm} & \hspace{-0.5mm} -1 \hspace{-0.5mm} & \hspace{-0.5mm} -1 \\
\end{pmatrix} \,.
\end{equation}
Thus, we have shown that it is possible to derive differential
equations of the type (\ref{eq:diff_eqs_single_ratio})
for the basis integrals in eq.~(\ref{eq:planar_DB_integral_basis})
and achieve a canonical form (\ref{eq:canonical_coeff_matrix_planar_DB}) of the system
without introducing integrals with doubled propagators in intermediate stages.

\subsection{Massless non-planar double-box diagram}

As a slightly more involved example, let us consider
the fully massless non-planar double-box diagram
shown in fig.~\ref{fig:massless_nonplanar_DB_z_variables}.

\begin{figure}[!h]
\begin{center}
\includegraphics[angle=0, width=0.28\textwidth]{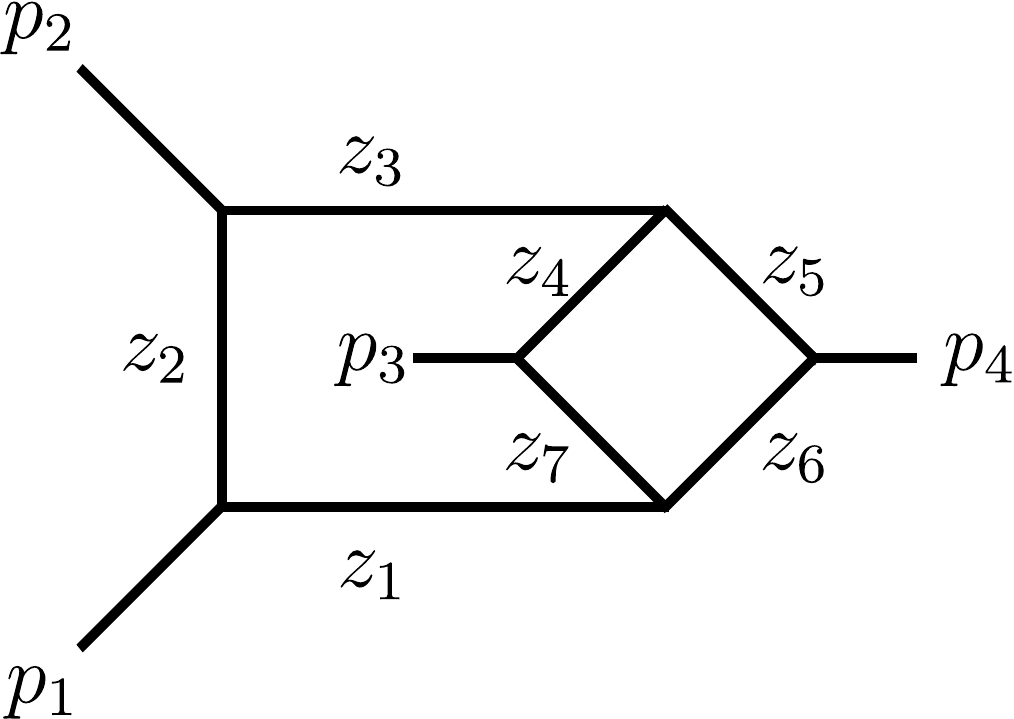}
{\vskip -0mm}
\caption{The fully massless non-planar double-box diagram.
All external momenta are taken to be outgoing.}
\label{fig:massless_nonplanar_DB_z_variables}
\end{center}
\end{figure}

\noindent We define the $z$-variables (\ref{eq:definition_of_z}) as follows,
\begin{equation}
\begin{alignedat}{3}
z_1&=\ell_1^2\,,           \hspace{3mm}  && z_2 = (\ell_1 - p_1)^2\,,  \hspace{3mm} && z_3 = (\ell_1 - P_{12})^2\,, \\
z_4&=(\ell_1 + \ell_2 + p_3)^2\,,  \hspace{3mm}  && z_5 = (\ell_2 - p_4)^2\,,  \hspace{3mm} && z_6 = \ell_2^2\,, \\
z_7&=(\ell_1 + \ell_2)^2\,,  \hspace{3mm}  && z_8 = (\ell_1 + p_4)^2\,,  \hspace{3mm} && z_9 = (\ell_2 + p_1)^2\,.
\end{alignedat}
\end{equation}
In terms of these variables and the Mandelstam invariants $s$ and $t$,
the Gram determinant in eq.~(\ref{eq:definition_of_Baikov_polynomial}) takes the following form,
\begin{equation}
F=\begin{vmatrix}
 0 \hspace{-0.5mm} & \hspace{-0.5mm} \frac{s}{2} \hspace{-0.5mm} & \hspace{-0.5mm} -\frac{s+t}{2} \hspace{-0.5mm} & \hspace{-0.5mm} \frac{z_1-z_2}{2} \hspace{-0.5mm} & \hspace{-0.5mm} \frac{z_9-z_6}{2} \\[2mm]
 \frac{s}{2} \hspace{-0.5mm} & \hspace{-0.5mm} 0 \hspace{-0.5mm} & \hspace{-0.5mm} \frac{t}{2} \hspace{-0.5mm} & \hspace{-0.5mm} \frac{z_2 - z_3 + s}{2} \hspace{-0.5mm} & \hspace{-0.5mm} \thead{\frac{1}{2}(z_3{-}z_4 \\ {+}z_5{+}z_7 \\ {-}z_8{-}z_9{-}s)} \\[5mm]
 -\frac{s+t}{2} \hspace{-0.5mm} & \hspace{-0.5mm} \frac{t}{2} \hspace{-0.5mm} & \hspace{-0.5mm} 0 \hspace{-0.5mm} & \hspace{-0.5mm} \frac{z_3 - z_8 - s}{2} \hspace{-0.5mm} & \hspace{-0.5mm} \thead{\frac{1}{2}(z_4{-}z_3 \\ {-}z_7{+}z_8{+}s)} \\[5mm]
 \frac{z_1 - z_2}{2} \hspace{-0.5mm} & \hspace{-0.5mm} \frac{z_2 - z_3 + s}{2} \hspace{-0.5mm} & \hspace{-0.5mm} \frac{z_3 - z_8 - s}{2} \hspace{-0.5mm} & \hspace{-0.5mm} z_1 \hspace{-0.5mm} & \hspace{-0.5mm} \frac{z_7 -z_1 - z_6}{2} \\[3mm]
 \frac{z_9 - z_6}{2} \hspace{-0.5mm} & \hspace{-0.5mm} \thead{\frac{1}{2}\big(z_3{-}z_4 \\ + z_5{+}z_7 \\ {-}z_8{-}z_9{-}s\big)}
 \hspace{-0.5mm} & \hspace{-0.5mm} \thead{\frac{1}{2} (z_4{-}z_3 \\ {-}z_7{+}z_8{+}s)} \hspace{-0.5mm} & \hspace{-0.5mm} \frac{z_7 - z_1 -z_6}{2} \hspace{-0.5mm} & \hspace{-0.5mm} z_6
\end{vmatrix} \,.
\end{equation}
Our aim is again to write down differential equations
for a basis of the vector space spanned by the diagram in
fig.~\ref{fig:massless_nonplanar_DB_z_variables} and its subdiagrams.
\textsc{Azurite} produces the following integral basis,
\begin{alignat}{3}
\boldsymbol{\mathcal{I}} = \Big( s^{-1 + 2\epsilon} I_{(1,0,0,1,1,0,0,0,0)}, & \hspace{3mm}
&& s^{-1 + 2\epsilon} I_{(0,1,0,1,0,1,0,0,0)}, \nonumber \\
s^{-1 + 2\epsilon} I_{(0,1,0,0,1,0,1,0,0)}, & \hspace{3mm}
&& s^{2\epsilon} I_{(1,0,1,1,0,1,0,0,0)}, \nonumber \\
s^{1 + 2\epsilon} I_{(1,1,1,1,0,1,0,0,0)}, & \hspace{3mm}
&& s^{1 + 2\epsilon} I_{(1,1,1,0,1,0,1,0,0)}, \nonumber \\
s^{1 + 2\epsilon} I_{(1,1,0,1,1,1,0,0,0)}, & \hspace{3mm}
&& s^{1 + 2\epsilon} I_{(1,1,0,1,1,0,1,0,0)}, \nonumber \\
s^{1 + 2\epsilon} I_{(0,1,0,1,1,1,1,0,0)}, & \hspace{3mm}
&& s^{2 + 2\epsilon} I_{(1,0,1,1,1,1,1,0,0)}, \nonumber \\
s^{2 + 2\epsilon} I_{(1,1,1,1,1,1,1,-1,0)}, & \hspace{3mm}
&& s^{3 + 2\epsilon} I_{(1,1,1,1,1,1,1,0,0)} \Big) \,,
\label{eq:non-planar_DB_integral_basis}
\end{alignat}
where we have again rescaled the basis integrals by
appropriate powers of $s$ to ensure that they only
depend on kinematics through $\chi$.

Our aim is now to write down differential equations for
the basis integrals in eq.~(\ref{eq:non-planar_DB_integral_basis})
without encountering squared propagators in intermediate stages.
We therefore ask whether the enhanced ideal membership
in eq.~(\ref{eq:Baikov_poly_ideal_membership_2}) holds.
Once again, the ``bonus'' ideal membership (\ref{eq:decomposed_derivative_of_Baikov_pol})
turns out to hold, with the cofactors,
\begin{align}
\mathbf{b} &=\left({\textstyle\frac{z_3 - z_8}{\chi (\chi +1) s}},
   {\textstyle\frac{z_3 - z_8 - \chi s - s}{\chi (\chi +1) s}},
   {\textstyle\frac{z_3 - z_8 - s}{\chi (\chi+1) s}},
   {\textstyle\frac{2z_3 - z_4 + z_7 - 2z_8 - s}{\chi (\chi +1) s}}, \right. \nonumber \\
&\hspace{7mm} \left.
   {\textstyle\frac{z_3 - z_4 + z_7 - z_8 - s}{\chi (\chi+1) s}},
   {\textstyle\frac{z_3 - z_4 + z_7 - z_8}{\chi (\chi +1) s}},
   {\textstyle\frac{2z_3 - z_4 + z_7 - 2z_8 - s}{\chi (\chi +1) s}}, \right. \nonumber \\
&\hspace{7mm} \left.
   {\textstyle\frac{z_3 - z_8 - s}{\chi (\chi +1) s}},
   {\textstyle\frac{z_3 - z_4 + z_7 - z_8 - \chi s - s}{\chi (\chi +1) s}} \right) \nonumber \\[2mm]
b &=  -{\textstyle \frac{4z_3 - 2z_4 + 2z_7 - 4z_8 - 2\chi s - 3s}{\chi (\chi +1) s}} \,.
\end{align}
We insert these expressions into eq.~(\ref{eq:numerator_of_differentiated_integrand})
to find the right-hand side of eq.~(\ref{eq:differential_equations}) and then
apply integration-by-parts reductions to the latter to find explicit
differential equations of the form (\ref{eq:diff_eqs_single_ratio}).

In order to present the resulting differential equations in a useful form, we use
\texttt{Fuchsia} to find a change-of-basis matrix $U$ (\ref{eq:change-of-basis_matrix})
to a canonical representation. In the case at hand,
we prefer to present the matrix inverse $U^{-1}$ which has a slightly
more compact expression. The latter takes the following form,
\begin{equation}
(U^{-1})_{ij} = \delta_{ij} v_{1,i} + \delta_{i,12} v_{2,j} \,,
\end{equation}
where the vectors $v_1$ and $v_2$ are given by,
\begin{align}
v_1 \hspace{0mm}&=\hspace{0mm} \Big(
\textstyle{\frac{30 \epsilon^3 (1 + 4\epsilon)}{(1 + \epsilon) (1 - 2\epsilon) (1 - 3\epsilon) (3\epsilon -2)}},
\textstyle{\frac{30 \epsilon^3 (1 + 4 \epsilon) (\chi +1)}{(1 + \epsilon) (1 - 2\epsilon) (1 - 3\epsilon) (2 - 3\epsilon)}}, \nonumber \\
&\hspace{-2mm}\textstyle{\frac{30 \epsilon^3 (1 + 4\epsilon) \chi}{(1 + \epsilon) (1 - 2\epsilon) (1 - 3\epsilon) (2 - 3\epsilon)}},
\textstyle{\frac{6\epsilon^2 (1 + 4\epsilon)}{(1 + \epsilon) (1 - 2\epsilon) (1 - 3\epsilon)}},
\textstyle{\frac{3\epsilon (1 + 4\epsilon)}{2 (1 + \epsilon) (-1 + 2 \epsilon)}}, \nonumber \\
&\hspace{-2mm}\textstyle{\frac{3\epsilon (1 + 4\epsilon)}{2 (1 + \epsilon) (-1 + 2 \epsilon)}},
-\textstyle{\frac{1 + 4\epsilon}{2 \chi (1 + \epsilon)}},
-\textstyle{\frac{1 + 4\epsilon}{2 (\chi +1) (1 + \epsilon)}},
\textstyle{\frac{1 + 4\epsilon}{2 (1 + \epsilon)}},
\textstyle{\frac{1 + 4\epsilon}{2 (1 + \epsilon)}}, \nonumber \\
&\hspace{-2mm}-\textstyle{\frac{1 + 4\epsilon}{2 (\chi +1) (1 + \epsilon)}},
-\textstyle{\frac{1 + 4 \epsilon}{2 \chi (1 + \epsilon)}} \Big) \,, \\[2mm]
v_2 \hspace{0mm}&=\hspace{0mm} \Big(
\textstyle{\frac{180 \epsilon}{\chi (\chi +1) (1 + \epsilon)}},
\textstyle{\frac{60}{\chi}},
\textstyle{\frac{60 (\chi + 1 + \chi \epsilon + 4\epsilon)}{\chi (\chi +1) (1 + \epsilon)}},
0, 0, 0, \nonumber \\
&-\textstyle{\frac{2 (\chi + 1 + \chi \epsilon + 4\epsilon)}{\chi (\chi +1) (1 + \epsilon)}},
-\textstyle{\frac{2}{\chi}},
-\textstyle{\frac{6\epsilon}{\chi (\chi +1) (1 + \epsilon)}},
0,
\textstyle{\frac{1 + 4\epsilon}{2 (\chi +1) (1 + \epsilon)}},
0
\Big) \,.
\end{align}
The transformed coefficient matrix takes the following form,
\begin{equation}
\widehat{A} = \epsilon \left( \frac{a_0}{\chi} + \frac{a_{-1}}{\chi + 1} \right) \,,
\label{eq:canonical_coeff_matrix_non-planar_DB}
\end{equation}
where $a_0$ and $a_{-1}$ are matrices with integer entries,
\begin{equation}
a_0 =
{\scriptsize
\left(
\begin{array}{cccccccccccc}
 0 \hspace{-0.5mm} & \hspace{-0.5mm} 0 \hspace{-0.5mm} & \hspace{-0.5mm} 0 \hspace{-0.5mm} & \hspace{-0.5mm} 0 \hspace{-0.5mm} & \hspace{-0.5mm} 0 \hspace{-0.5mm} & \hspace{-0.5mm} 0 \hspace{-0.5mm} & \hspace{-0.5mm} 0 \hspace{-0.5mm} & \hspace{-0.5mm} 0 \hspace{-0.5mm} & \hspace{-0.5mm} 0 \hspace{-0.5mm} & \hspace{-0.5mm} 0 \hspace{-0.5mm} & \hspace{-0.5mm} 0 \hspace{-0.5mm} & \hspace{-0.5mm} 0 \\
 0 \hspace{-0.5mm} & \hspace{-0.5mm} 0 \hspace{-0.5mm} & \hspace{-0.5mm} 0 \hspace{-0.5mm} & \hspace{-0.5mm} 0 \hspace{-0.5mm} & \hspace{-0.5mm} 0 \hspace{-0.5mm} & \hspace{-0.5mm} 0 \hspace{-0.5mm} & \hspace{-0.5mm} 0 \hspace{-0.5mm} & \hspace{-0.5mm} 0 \hspace{-0.5mm} & \hspace{-0.5mm} 0 \hspace{-0.5mm} & \hspace{-0.5mm} 0 \hspace{-0.5mm} & \hspace{-0.5mm} 0 \hspace{-0.5mm} & \hspace{-0.5mm} 0 \\
 0 \hspace{-0.5mm} & \hspace{-0.5mm} 0 \hspace{-0.5mm} & \hspace{-0.5mm} -2 \hspace{-0.5mm} & \hspace{-0.5mm} 0 \hspace{-0.5mm} & \hspace{-0.5mm} 0 \hspace{-0.5mm} & \hspace{-0.5mm} 0 \hspace{-0.5mm} & \hspace{-0.5mm} 0 \hspace{-0.5mm} & \hspace{-0.5mm} 0 \hspace{-0.5mm} & \hspace{-0.5mm} 0 \hspace{-0.5mm} & \hspace{-0.5mm} 0 \hspace{-0.5mm} & \hspace{-0.5mm} 0 \hspace{-0.5mm} & \hspace{-0.5mm} 0 \\
 0 \hspace{-0.5mm} & \hspace{-0.5mm} 0 \hspace{-0.5mm} & \hspace{-0.5mm} 0 \hspace{-0.5mm} & \hspace{-0.5mm} 0 \hspace{-0.5mm} & \hspace{-0.5mm} 0 \hspace{-0.5mm} & \hspace{-0.5mm} 0 \hspace{-0.5mm} & \hspace{-0.5mm} 0 \hspace{-0.5mm} & \hspace{-0.5mm} 0 \hspace{-0.5mm} & \hspace{-0.5mm} 0 \hspace{-0.5mm} & \hspace{-0.5mm} 0 \hspace{-0.5mm} & \hspace{-0.5mm} 0 \hspace{-0.5mm} & \hspace{-0.5mm} 0 \\
 0 \hspace{-0.5mm} & \hspace{-0.5mm} 20 \hspace{-0.5mm} & \hspace{-0.5mm} 0 \hspace{-0.5mm} & \hspace{-0.5mm} 4 \hspace{-0.5mm} & \hspace{-0.5mm} 1 \hspace{-0.5mm} & \hspace{-0.5mm} 0 \hspace{-0.5mm} & \hspace{-0.5mm} 0 \hspace{-0.5mm} & \hspace{-0.5mm} 0 \hspace{-0.5mm} & \hspace{-0.5mm} 0 \hspace{-0.5mm} & \hspace{-0.5mm} 0 \hspace{-0.5mm} & \hspace{-0.5mm} 0 \hspace{-0.5mm} & \hspace{-0.5mm} 0 \\
 0 \hspace{-0.5mm} & \hspace{-0.5mm} 0 \hspace{-0.5mm} & \hspace{-0.5mm} 20 \hspace{-0.5mm} & \hspace{-0.5mm} -4 \hspace{-0.5mm} & \hspace{-0.5mm} 0 \hspace{-0.5mm} & \hspace{-0.5mm} -2 \hspace{-0.5mm} & \hspace{-0.5mm} 0 \hspace{-0.5mm} & \hspace{-0.5mm} 0 \hspace{-0.5mm} & \hspace{-0.5mm} 0 \hspace{-0.5mm} & \hspace{-0.5mm} 0 \hspace{-0.5mm} & \hspace{-0.5mm} 0 \hspace{-0.5mm} & \hspace{-0.5mm} 0 \\
 0 \hspace{-0.5mm} & \hspace{-0.5mm} 0 \hspace{-0.5mm} & \hspace{-0.5mm} 0 \hspace{-0.5mm} & \hspace{-0.5mm} 0 \hspace{-0.5mm} & \hspace{-0.5mm} 0 \hspace{-0.5mm} & \hspace{-0.5mm} 0 \hspace{-0.5mm} & \hspace{-0.5mm} 2 \hspace{-0.5mm} & \hspace{-0.5mm} 0 \hspace{-0.5mm} & \hspace{-0.5mm} 0 \hspace{-0.5mm} & \hspace{-0.5mm} 0 \hspace{-0.5mm} & \hspace{-0.5mm} 0 \hspace{-0.5mm} & \hspace{-0.5mm} 0 \\
 -60 \hspace{-0.5mm} & \hspace{-0.5mm} 0 \hspace{-0.5mm} & \hspace{-0.5mm} -60 \hspace{-0.5mm} & \hspace{-0.5mm} 0 \hspace{-0.5mm} & \hspace{-0.5mm} 0 \hspace{-0.5mm} & \hspace{-0.5mm} 0 \hspace{-0.5mm} & \hspace{-0.5mm} 0 \hspace{-0.5mm} & \hspace{-0.5mm} -2 \hspace{-0.5mm} & \hspace{-0.5mm} 0 \hspace{-0.5mm} & \hspace{-0.5mm} 0 \hspace{-0.5mm} & \hspace{-0.5mm} 0 \hspace{-0.5mm} & \hspace{-0.5mm} 0 \\
 0 \hspace{-0.5mm} & \hspace{-0.5mm} -60 \hspace{-0.5mm} & \hspace{-0.5mm} -60 \hspace{-0.5mm} & \hspace{-0.5mm} 0 \hspace{-0.5mm} & \hspace{-0.5mm} 0 \hspace{-0.5mm} & \hspace{-0.5mm} 0 \hspace{-0.5mm} & \hspace{-0.5mm} 0 \hspace{-0.5mm} & \hspace{-0.5mm} 0 \hspace{-0.5mm} & \hspace{-0.5mm} -2 \hspace{-0.5mm} & \hspace{-0.5mm} 0 \hspace{-0.5mm} & \hspace{-0.5mm} 0 \hspace{-0.5mm} & \hspace{-0.5mm} 0 \\
 0 \hspace{-0.5mm} & \hspace{-0.5mm} 0 \hspace{-0.5mm} & \hspace{-0.5mm} 0 \hspace{-0.5mm} & \hspace{-0.5mm} 0 \hspace{-0.5mm} & \hspace{-0.5mm} 0 \hspace{-0.5mm} & \hspace{-0.5mm} 0 \hspace{-0.5mm} & \hspace{-0.5mm} 0 \hspace{-0.5mm} & \hspace{-0.5mm} 0 \hspace{-0.5mm} & \hspace{-0.5mm} 0 \hspace{-0.5mm} & \hspace{-0.5mm} 0 \hspace{-0.5mm} & \hspace{-0.5mm} 0 \hspace{-0.5mm} & \hspace{-0.5mm} 0 \\
 -360 \hspace{-0.5mm} & \hspace{-0.5mm} -60 \hspace{-0.5mm} & \hspace{-0.5mm} 660 \hspace{-0.5mm} & \hspace{-0.5mm} -36 \hspace{-0.5mm} & \hspace{-0.5mm} 0 \hspace{-0.5mm} & \hspace{-0.5mm} -18 \hspace{-0.5mm} & \hspace{-0.5mm} -4 \hspace{-0.5mm} & \hspace{-0.5mm} -16 \hspace{-0.5mm} & \hspace{-0.5mm} -6 \hspace{-0.5mm} & \hspace{-0.5mm} 1 \hspace{-0.5mm} & \hspace{-0.5mm} 1 \hspace{-0.5mm} & \hspace{-0.5mm} -1 \\
 -120 \hspace{-0.5mm} & \hspace{-0.5mm} 420 \hspace{-0.5mm} & \hspace{-0.5mm} 420 \hspace{-0.5mm} & \hspace{-0.5mm} -36 \hspace{-0.5mm} & \hspace{-0.5mm} 0 \hspace{-0.5mm} & \hspace{-0.5mm} -18 \hspace{-0.5mm} & \hspace{-0.5mm} -16 \hspace{-0.5mm} & \hspace{-0.5mm} -12 \hspace{-0.5mm} & \hspace{-0.5mm} 6 \hspace{-0.5mm} & \hspace{-0.5mm} 0 \hspace{-0.5mm} & \hspace{-0.5mm} 0 \hspace{-0.5mm} & \hspace{-0.5mm} -2 \\
\end{array}
\right) \,,
}
\end{equation}
and
\begin{equation}
a_{-1}=
{\scriptsize
\left(
\begin{array}{cccccccccccc}
 0 \hspace{-0.5mm} & \hspace{-0.5mm} 0 \hspace{-0.5mm} & \hspace{-0.5mm} 0 \hspace{-0.5mm} & \hspace{-0.5mm} 0 \hspace{-0.5mm} & \hspace{-0.5mm} 0 \hspace{-0.5mm} & \hspace{-0.5mm} 0 \hspace{-0.5mm} & \hspace{-0.5mm} 0 \hspace{-0.5mm} & \hspace{-0.5mm} 0 \hspace{-0.5mm} & \hspace{-0.5mm} 0 \hspace{-0.5mm} & \hspace{-0.5mm} 0 \hspace{-0.5mm} & \hspace{-0.5mm} 0 \hspace{-0.5mm} & \hspace{-0.5mm} 0 \\
 0 \hspace{-0.5mm} & \hspace{-0.5mm} -2 \hspace{-0.5mm} & \hspace{-0.5mm} 0 \hspace{-0.5mm} & \hspace{-0.5mm} 0 \hspace{-0.5mm} & \hspace{-0.5mm} 0 \hspace{-0.5mm} & \hspace{-0.5mm} 0 \hspace{-0.5mm} & \hspace{-0.5mm} 0 \hspace{-0.5mm} & \hspace{-0.5mm} 0 \hspace{-0.5mm} & \hspace{-0.5mm} 0 \hspace{-0.5mm} & \hspace{-0.5mm} 0 \hspace{-0.5mm} & \hspace{-0.5mm} 0 \hspace{-0.5mm} & \hspace{-0.5mm} 0 \\
 0 \hspace{-0.5mm} & \hspace{-0.5mm} 0 \hspace{-0.5mm} & \hspace{-0.5mm} 0 \hspace{-0.5mm} & \hspace{-0.5mm} 0 \hspace{-0.5mm} & \hspace{-0.5mm} 0 \hspace{-0.5mm} & \hspace{-0.5mm} 0 \hspace{-0.5mm} & \hspace{-0.5mm} 0 \hspace{-0.5mm} & \hspace{-0.5mm} 0 \hspace{-0.5mm} & \hspace{-0.5mm} 0 \hspace{-0.5mm} & \hspace{-0.5mm} 0 \hspace{-0.5mm} & \hspace{-0.5mm} 0 \hspace{-0.5mm} & \hspace{-0.5mm} 0 \\
 0 \hspace{-0.5mm} & \hspace{-0.5mm} 0 \hspace{-0.5mm} & \hspace{-0.5mm} 0 \hspace{-0.5mm} & \hspace{-0.5mm} 0 \hspace{-0.5mm} & \hspace{-0.5mm} 0 \hspace{-0.5mm} & \hspace{-0.5mm} 0 \hspace{-0.5mm} & \hspace{-0.5mm} 0 \hspace{-0.5mm} & \hspace{-0.5mm} 0 \hspace{-0.5mm} & \hspace{-0.5mm} 0 \hspace{-0.5mm} & \hspace{-0.5mm} 0 \hspace{-0.5mm} & \hspace{-0.5mm} 0 \hspace{-0.5mm} & \hspace{-0.5mm} 0 \\
 0 \hspace{-0.5mm} & \hspace{-0.5mm} -20 \hspace{-0.5mm} & \hspace{-0.5mm} 0 \hspace{-0.5mm} & \hspace{-0.5mm} -4 \hspace{-0.5mm} & \hspace{-0.5mm} -2 \hspace{-0.5mm} & \hspace{-0.5mm} 0 \hspace{-0.5mm} & \hspace{-0.5mm} 0 \hspace{-0.5mm} & \hspace{-0.5mm} 0 \hspace{-0.5mm} & \hspace{-0.5mm} 0 \hspace{-0.5mm} & \hspace{-0.5mm} 0 \hspace{-0.5mm} & \hspace{-0.5mm} 0 \hspace{-0.5mm} & \hspace{-0.5mm} 0 \\
 0 \hspace{-0.5mm} & \hspace{-0.5mm} 0 \hspace{-0.5mm} & \hspace{-0.5mm} -20 \hspace{-0.5mm} & \hspace{-0.5mm} 4 \hspace{-0.5mm} & \hspace{-0.5mm} 0 \hspace{-0.5mm} & \hspace{-0.5mm} 1 \hspace{-0.5mm} & \hspace{-0.5mm} 0 \hspace{-0.5mm} & \hspace{-0.5mm} 0 \hspace{-0.5mm} & \hspace{-0.5mm} 0 \hspace{-0.5mm} & \hspace{-0.5mm} 0 \hspace{-0.5mm} & \hspace{-0.5mm} 0 \hspace{-0.5mm} & \hspace{-0.5mm} 0 \\
 60 \hspace{-0.5mm} & \hspace{-0.5mm} -60 \hspace{-0.5mm} & \hspace{-0.5mm} 0 \hspace{-0.5mm} & \hspace{-0.5mm} 0 \hspace{-0.5mm} & \hspace{-0.5mm} 0 \hspace{-0.5mm} & \hspace{-0.5mm} 0 \hspace{-0.5mm} & \hspace{-0.5mm} -2 \hspace{-0.5mm} & \hspace{-0.5mm} 0 \hspace{-0.5mm} & \hspace{-0.5mm} 0 \hspace{-0.5mm} & \hspace{-0.5mm} 0 \hspace{-0.5mm} & \hspace{-0.5mm} 0 \hspace{-0.5mm} & \hspace{-0.5mm} 0 \\
 0 \hspace{-0.5mm} & \hspace{-0.5mm} 0 \hspace{-0.5mm} & \hspace{-0.5mm} 0 \hspace{-0.5mm} & \hspace{-0.5mm} 0 \hspace{-0.5mm} & \hspace{-0.5mm} 0 \hspace{-0.5mm} & \hspace{-0.5mm} 0 \hspace{-0.5mm} & \hspace{-0.5mm} 0 \hspace{-0.5mm} & \hspace{-0.5mm} 2 \hspace{-0.5mm} & \hspace{-0.5mm} 0 \hspace{-0.5mm} & \hspace{-0.5mm} 0 \hspace{-0.5mm} & \hspace{-0.5mm} 0 \hspace{-0.5mm} & \hspace{-0.5mm} 0 \\
 0 \hspace{-0.5mm} & \hspace{-0.5mm} 60 \hspace{-0.5mm} & \hspace{-0.5mm} 60 \hspace{-0.5mm} & \hspace{-0.5mm} 0 \hspace{-0.5mm} & \hspace{-0.5mm} 0 \hspace{-0.5mm} & \hspace{-0.5mm} 0 \hspace{-0.5mm} & \hspace{-0.5mm} 0 \hspace{-0.5mm} & \hspace{-0.5mm} 0 \hspace{-0.5mm} & \hspace{-0.5mm} -2 \hspace{-0.5mm} & \hspace{-0.5mm} 0 \hspace{-0.5mm} & \hspace{-0.5mm} 0 \hspace{-0.5mm} & \hspace{-0.5mm} 0 \\
 0 \hspace{-0.5mm} & \hspace{-0.5mm} 0 \hspace{-0.5mm} & \hspace{-0.5mm} 0 \hspace{-0.5mm} & \hspace{-0.5mm} 0 \hspace{-0.5mm} & \hspace{-0.5mm} 0 \hspace{-0.5mm} & \hspace{-0.5mm} 0 \hspace{-0.5mm} & \hspace{-0.5mm} 0 \hspace{-0.5mm} & \hspace{-0.5mm} 0 \hspace{-0.5mm} & \hspace{-0.5mm} 0 \hspace{-0.5mm} & \hspace{-0.5mm} 0 \hspace{-0.5mm} & \hspace{-0.5mm} 0 \hspace{-0.5mm} & \hspace{-0.5mm} 0 \\
 -360 \hspace{-0.5mm} & \hspace{-0.5mm} -180 \hspace{-0.5mm} & \hspace{-0.5mm} -180 \hspace{-0.5mm} & \hspace{-0.5mm} -36 \hspace{-0.5mm} & \hspace{-0.5mm} -18 \hspace{-0.5mm} & \hspace{-0.5mm} 0 \hspace{-0.5mm} & \hspace{-0.5mm} 12 \hspace{-0.5mm} & \hspace{-0.5mm} 0 \hspace{-0.5mm} & \hspace{-0.5mm} 6 \hspace{-0.5mm} & \hspace{-0.5mm} 0 \hspace{-0.5mm} & \hspace{-0.5mm} -2 \hspace{-0.5mm} & \hspace{-0.5mm} 0 \\
 -600 \hspace{-0.5mm} & \hspace{-0.5mm} -660 \hspace{-0.5mm} & \hspace{-0.5mm} 60 \hspace{-0.5mm} & \hspace{-0.5mm} -36 \hspace{-0.5mm} & \hspace{-0.5mm} -18 \hspace{-0.5mm} & \hspace{-0.5mm} 0 \hspace{-0.5mm} & \hspace{-0.5mm} 24 \hspace{-0.5mm} & \hspace{-0.5mm} -4 \hspace{-0.5mm} & \hspace{-0.5mm} -6 \hspace{-0.5mm} & \hspace{-0.5mm} 1 \hspace{-0.5mm} & \hspace{-0.5mm} -1 \hspace{-0.5mm} & \hspace{-0.5mm} 1 \\
\end{array}
\right) \,.
}
\end{equation}
We conclude that it is possible to derive differential
equations for the basis integrals in eq.~(\ref{eq:non-planar_DB_integral_basis})
of the type (\ref{eq:diff_eqs_single_ratio}) and achieve
a canonical form (\ref{eq:canonical_coeff_matrix_non-planar_DB}) of the system
without introducing integrals with doubled propagators in intermediate stages.

\subsection{Further verification of the enhanced ideal membership}

We have verified that the enhanced ideal membership
in eq.~(\ref{eq:Baikov_poly_ideal_membership_2}) holds
in many multi-loop examples. Rather than writing down the
explicit form of the resulting differential equations,
in this section we restrict ourselves to displaying
in fig.~\ref{fig:check_of_enhanced_membership}
a selection of diagrams for which eq.~(\ref{eq:Baikov_poly_ideal_membership_2})
was verified to hold.

As explained above, we note that eq.~(\ref{eq:Baikov_poly_ideal_membership_2})
implies that we can derive differential equations of the type (\ref{eq:diff_eqs_schematic})
without introducing integrals with doubled propagators in intermediate stages.

\begin{figure}[!h]
\begin{center}
\includegraphics[angle=0, width=0.4\textwidth]{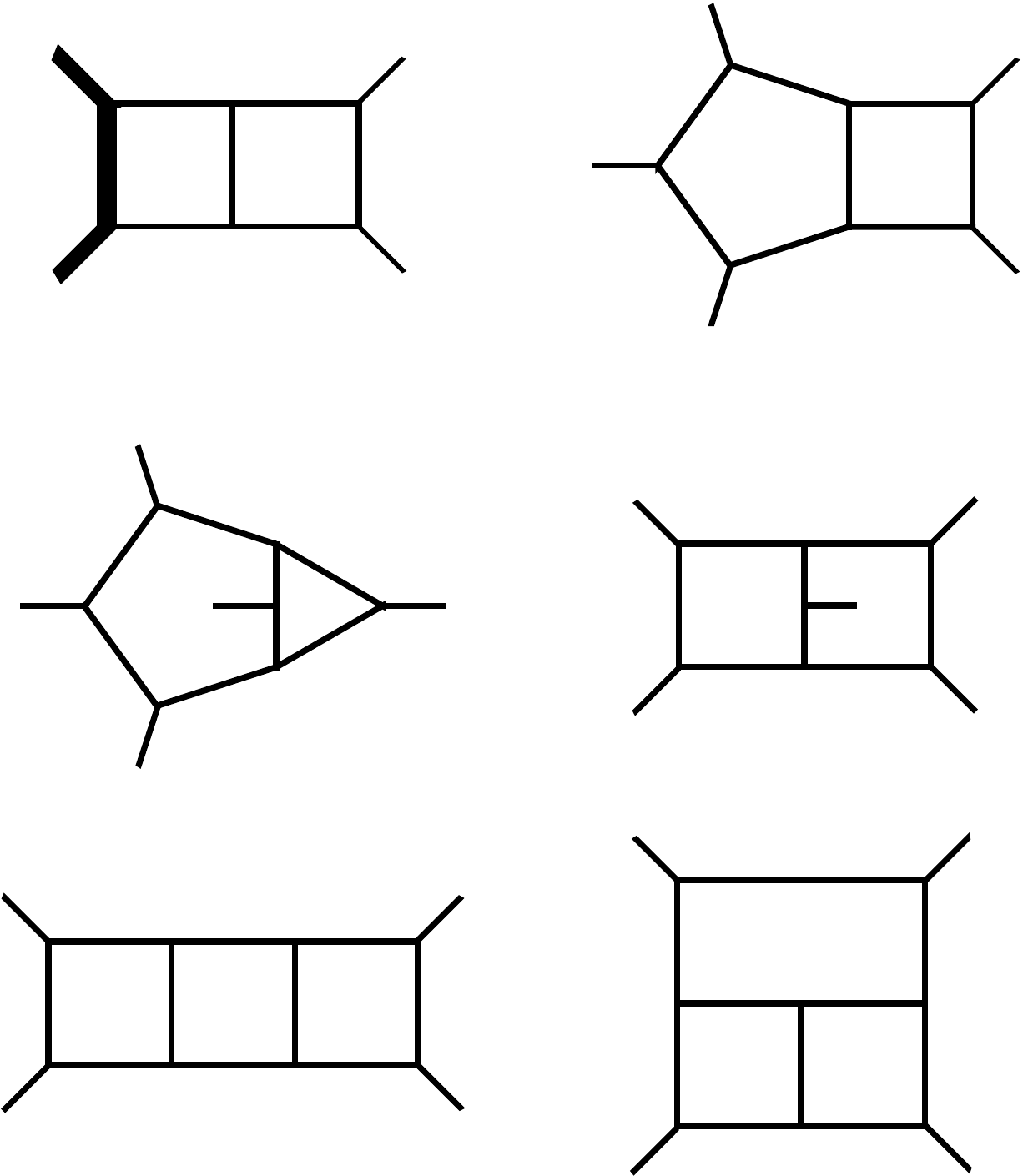}
{\vskip -0mm}
\caption{A selection of diagrams for which the enhanced
ideal membership in eq.~(\ref{eq:Baikov_poly_ideal_membership_2}) has been verified.
The bold lines represent massive momenta and propagators.}
\label{fig:check_of_enhanced_membership}
\end{center}
\end{figure}

\subsection{Counterexample of the enhanced ideal membership}

Although the enhanced ideal membership
in eq.~(\ref{eq:Baikov_poly_ideal_membership_2}) holds true
for a large class of Feynman diagrams, we have found that
it is not a general property of the Baikov polynomial $F$.

In particular, consider the diagram displayed in
fig.~\ref{fig:counterexample}. Upon computing a
Gr{\"o}bner basis $\mathcal{G}$ of the ideal and
performing polynomial division of $\frac{\partial F}{\partial \chi}$
with respect to $\mathcal{G}$, one finds a nonvanishing
remainder. Hence, the diagram in
fig.~\ref{fig:counterexample} provides a counterexample
to the enhanced ideal membership
in eq.~(\ref{eq:Baikov_poly_ideal_membership_2}).

\begin{figure}[!h]
\begin{center}
\includegraphics[angle=0, width=0.25\textwidth]{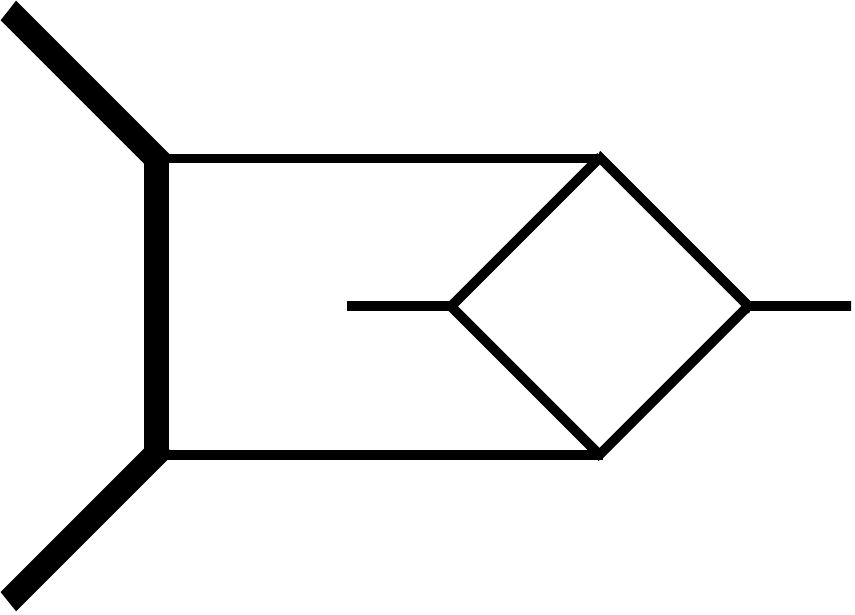}
{\vskip -0mm}
\caption{Non-planar double-box diagram. The bold lines represent massive
momenta and propagators. For this diagram, the enhanced ideal membership
eq.~(\ref{eq:Baikov_poly_ideal_membership_2}) does not hold.}
\label{fig:counterexample}
\end{center}
\end{figure}

\section{Conclusions}\label{sec:conclusions}

In this paper we have studied differential equations for loop integrals
in Baikov representation. We gave a proof that the Baikov polynomial satisfies
the ideal membership property (\ref{eq:Baikov_poly_ideal_membership_1}).
We have shown that, as a result, differential equations can be derived
without involving dimension-shift identities. This is advantageous
because the latter are computationally intensive to generate,
as they require integration-by-parts identities that involve
integrals with high powers of loop momentum monomials.
We remark that the proof is constructive
and gives an explicit construction of the cofactors.

We have moreover shown that an enhanced ideal membership property
(\ref{eq:Baikov_poly_ideal_membership_2}) holds for a large class of
multi-loop diagrams. As a result, differential equations
of the type (\ref{eq:diff_eqs_schematic}) can be derived
without introducing integrals with doubled propagators in intermediate stages.
This is beneficial because it limits the set of integrals
for which integration-by-parts reductions must be determined.

At the same time, we have identified a counterexample to
eq.~(\ref{eq:Baikov_poly_ideal_membership_2}), showing that it
is not a general property of the Baikov polynomial. An interesting
open problem is therefore to classify the diagrams for which
the enhanced ideal membership property (\ref{eq:Baikov_poly_ideal_membership_2})
holds.

We emphasize that eq.~(\ref{eq:Baikov_poly_ideal_membership_2}) holds
for highly non-trivial loop diagrams whose differential equations
are not attainable with standard methods. It therefore appears
to be a promising tool for the calculation of multi-loop integrals.

{\bf Acknowledgments}

We thank Roman N. Lee for collaboration at an early stage of this work and very useful discussions.
We also thank Hjalte~Frellesvig, Harald~Ita, David~A.~Kosower, Erik~Panzer,
Costas~Papadopoulos, Vladimir~A.~Smirnov and Mao Zeng for useful discussions.
The research leading to these results has received
funding from Swiss National Science Foundation (Ambizione grant PZ00P2
161341) and funding from the European Research Council (ERC) under the European Union’s Horizon 2020 research and innovation programme (grant agreement No 725110).
The work of YZ is also partially supported by the
Swiss National Science Foundation through the NCCR SwissMap, Grant number 141869.
The work of JB is supported by grant no. 615203 from the European Research Council
under the FP7 and by the Swiss National Science Foundation
through the NCCR SwissMAP.
The work of KJL is supported by ERC-2014-CoG, Grant number 648630 IQFT.

\appendix
\numberwithin{equation}{section}
\renewcommand{\theequation}{\Alph{section}.\arabic{equation}}

\section{Differential equations from dimension-shift identities}\label{sec:diff_eqs_via_dimension_shifts}

In this Appendix we work out the form of the differential equations
(\ref{eq:diff_eqs_in_Baikov_rep_1}) when $\frac{\partial F}{\partial \chi}$
is left unchanged. As discussed below this equation,
the $\frac{1}{F}$ factor in the second line has the effect
of shifting the space-time dimension from $D$ to $D-2$.
Accordingly, we will show how to derive dimension-shift identities
relating the latter two. We remark that in this approach,
integrals with squared propagators are not generated
in intermediate stages. It is therefore entirely compatible
with the IBP reduction formalisms of
refs.~\cite{Gluza:2010ws,Ita:2015tya,Larsen:2015ped,Bern:2017gdk}.

It will be more convenient for the presentation to consider
a $(D+2)$-dimensional integral. Taking the Baikov representation
(\ref{eq:Baikov_representation}) of the latter and applying
a derivative with respect to any external invariant $\chi$ we find,
\begin{align}
& \hspace{-10mm} \frac{\partial}{\partial \chi} I_j(N_j; \alpha; D + 2) \nonumber \\
= &\hspace{1mm} \frac{2^{L-1} \Gamma(D-L-E+2)}{\Gamma(D-E+1) \hspace{0.3mm} U}
\hspace{0.6mm} I_j \hspace{-0.6mm}\left(N_j \frac{\partial F}{\partial \chi}; \alpha; D\right) \nonumber \\[1mm]
& + \frac{E-D-1}{2 U} \frac{\partial U}{\partial \chi} I_j(N_j; \alpha; D + 2) \,.
\label{eq:derivative_of_(D+2)-dim_integral}
\end{align}
We remark that, in computing the ratio $\frac{C_E^L(D+2)}{C_E^L(D)}$ above,
we used the fact that the entries of $A$ (defined in eq.~(\ref{eq:relation_of_z_to_x}))
are integers which are independent of $D$. Hence $\det A$ cancels out of
the ratio, cf.~eq.~(\ref{eq:Baikov_prefactor}).

Our aim is now to re-express the $(D+2)$-dimensional integrals in
eq.~(\ref{eq:derivative_of_(D+2)-dim_integral}) in terms of $D$-dimensional
integrals. To determine the requisite dimension-shift identities \cite{Bern:1992em,Bern:1993kr,Tarasov:1996br,Lee:2009dh,Lee:2010wea},
we start by observing that,
\begin{align}
I_j (N_j; \alpha; D + 2) &= \frac{C_E^L(D+2)}{C_E^L(D) \hspace{0.3mm} U} I_j (F N_j; \alpha; D) \\
&= \frac{2^L \Gamma(D-L-E+1)}{\Gamma(D-E+1) \hspace{0.3mm} U} I_j (F N_j; \alpha; D) \,.
\label{eq:D-dim_shift_identity_before_IBP}
\end{align}
Now, by applying integration-by-parts reductions to the right-hand side of eq.~(\ref{eq:D-dim_shift_identity_before_IBP}),
we can re-express $I_j (F N_j; \alpha; D)$ as a linear combination of the basis integrals.
In this way we find,
\begin{equation}
I_j (N_j; \alpha; D + 2) \hspace{0.6mm}=\hspace{0.6mm} \sum_{k=1}^M T_{jk} (\chi) I_k (N_k; \alpha; D) \,.
\label{eq:D-dim_shift_identities}
\end{equation}
Applying $\frac{\partial}{\partial \chi}$ to eq.~(\ref{eq:D-dim_shift_identities}) and comparing
the resulting expression for $\frac{\partial}{\partial \chi} I_j (N_j; \alpha; D + 2)$ with
eq.~(\ref{eq:derivative_of_(D+2)-dim_integral}) we find, after applying eq.~(\ref{eq:D-dim_shift_identities})
and multiplying by $(T^{-1})_{ij}$ from the left,
\begin{equation}
\frac{\partial}{\partial \chi} I_i (N_i; \alpha; D) = \sum_{k=1}^M A_{ik} I_k (N_k; \alpha; D) \,,
\end{equation}
where the coefficient matrix is given by,
\begin{align}
A_{ik} \hspace{0mm}&=\hspace{0mm} \sum_{j=1}^M
\big(T^{-1}\big)_{ij} \left[ \frac{2^{L-1} \Gamma(D{-}L{-}E{+}2)}
{\Gamma(D{-}E{+}1) \hspace{0.3mm} U} R_{jk} - \frac{\partial T_{jk} (\chi)}{\partial \chi} \right] \nonumber \\
& \hspace{6mm} + \frac{E-D-1}{2 U} \frac{\partial U}{\partial \chi} \delta_{ik} \,,
\end{align}
and where the matrix $R$ arises as the matrix of IBP reduction coefficients,
\begin{equation}
I_j \left( N_j \frac{\partial F}{\partial \chi}; \alpha; D \right)
= \sum_{k=1}^M R_{jk} I_k (N_k; \alpha; D) \,.
\label{eq:IBP_reduction_of_insertion_of_derivative_of_Baikov_pol}
\end{equation}
The above dimension-shift approach has the virtue that
it avoids generating integrals with squared propagators
in intermediate stages. It is therefore amenable to the
the IBP reduction formalisms of
refs.~\cite{Gluza:2010ws,Ita:2015tya,Larsen:2015ped,Bern:2017gdk}.
However, performing the integration-by-parts reductions
of the right-hand side of eq.~(\ref{eq:IBP_reduction_of_insertion_of_derivative_of_Baikov_pol})
is computationally intensive in practice. This is because the
Baikov polynomial $F$ has higher powers of monomials than
encountered in the Feynman rules of gauge theory, and as a result
Gaussian elimination must be applied to large linear systems.

\bibliographystyle{h-physrev}

\bibliography{Diff_eqs_in_Baikov_rep}

\begin{thebibliography}{10}

\bibitem{Bonciani:2008az}
R.~Bonciani, A.~Ferroglia, T.~Gehrmann, D.~Maitre, and C.~Studerus,
\newblock JHEP {\bf 07}, 129 (2008), 0806.2301.

\bibitem{Gehrmann:2009vu}
T.~Gehrmann and E.~W.~N. Glover,
\newblock Phys. Lett. {\bf B676}, 146 (2009), 0904.2665.

\bibitem{Bonciani:2009nb}
R.~Bonciani, A.~Ferroglia, T.~Gehrmann, and C.~Studerus,
\newblock JHEP {\bf 08}, 067 (2009), 0906.3671.

\bibitem{Bonciani:2010mn}
R.~Bonciani, A.~Ferroglia, T.~Gehrmann, A.~von Manteuffel, and C.~Studerus,
\newblock JHEP {\bf 01}, 102 (2011), 1011.6661.

\bibitem{Gehrmann:2011ab}
T.~Gehrmann and L.~Tancredi,
\newblock JHEP {\bf 02}, 004 (2012), 1112.1531.

\bibitem{Gehrmann:2011aa}
T.~Gehrmann, M.~Jaquier, E.~W.~N. Glover, and A.~Koukoutsakis,
\newblock JHEP {\bf 02}, 056 (2012), 1112.3554.

\bibitem{Gehrmann:2013vga}
T.~Gehrmann, L.~Tancredi, and E.~Weihs,
\newblock JHEP {\bf 04}, 101 (2013), 1302.2630.

\bibitem{vonManteuffel:2013uoa}
A.~von Manteuffel and C.~Studerus,
\newblock JHEP {\bf 10}, 037 (2013), 1306.3504.

\bibitem{Gehrmann:2013cxs}
T.~Gehrmann, L.~Tancredi, and E.~Weihs,
\newblock JHEP {\bf 08}, 070 (2013), 1306.6344.

\bibitem{Henn:2013woa}
J.~M. Henn and V.~A. Smirnov,
\newblock JHEP {\bf 11}, 041 (2013), 1307.4083.

\bibitem{Bonciani:2013ywa}
R.~Bonciani, A.~Ferroglia, T.~Gehrmann, A.~von Manteuffel, and C.~Studerus,
\newblock JHEP {\bf 12}, 038 (2013), 1309.4450.

\bibitem{Henn:2014lfa}
J.~M. Henn, K.~Melnikov, and V.~A. Smirnov,
\newblock JHEP {\bf 05}, 090 (2014), 1402.7078.

\bibitem{Gehrmann:2014bfa}
T.~Gehrmann, A.~von Manteuffel, L.~Tancredi, and E.~Weihs,
\newblock JHEP {\bf 06}, 032 (2014), 1404.4853.

\bibitem{Caola:2014lpa}
F.~Caola, J.~M. Henn, K.~Melnikov, and V.~A. Smirnov,
\newblock JHEP {\bf 09}, 043 (2014), 1404.5590.

\bibitem{Caola:2014iua}
F.~Caola, J.~M. Henn, K.~Melnikov, A.~V. Smirnov, and V.~A. Smirnov,
\newblock JHEP {\bf 11}, 041 (2014), 1408.6409.

\bibitem{Papadopoulos:2014hla}
C.~G. Papadopoulos, D.~Tommasini, and C.~Wever,
\newblock JHEP {\bf 01}, 072 (2015), 1409.6114.

\bibitem{Huber:2015bva}
T.~Huber and S.~Kr{\"a}nkl,
\newblock JHEP {\bf 04}, 140 (2015), 1503.00735.

\bibitem{Gehrmann:2015ora}
T.~Gehrmann, A.~von Manteuffel, and L.~Tancredi,
\newblock JHEP {\bf 09}, 128 (2015), 1503.04812.

\bibitem{Caola:2015ila}
F.~Caola, J.~M. Henn, K.~Melnikov, A.~V. Smirnov, and V.~A. Smirnov,
\newblock JHEP {\bf 06}, 129 (2015), 1503.08759.

\bibitem{vonManteuffel:2015msa}
A.~von Manteuffel and L.~Tancredi,
\newblock JHEP {\bf 06}, 197 (2015), 1503.08835.

\bibitem{Bonciani:2016ypc}
R.~Bonciani, S.~Di~Vita, P.~Mastrolia, and U.~Schubert,
\newblock JHEP {\bf 09}, 091 (2016), 1604.08581.

\bibitem{Bonciani:2016qxi}
R.~Bonciani {\em et~al.},
\newblock JHEP {\bf 12}, 096 (2016), 1609.06685.

\bibitem{Melnikov:2017pgf}
K.~Melnikov, L.~Tancredi, and C.~Wever,
\newblock Phys. Rev. {\bf D95}, 054012 (2017), 1702.00426.

\bibitem{Mastrolia:2017pfy}
P.~Mastrolia, M.~Passera, A.~Primo, and U.~Schubert,
\newblock JHEP {\bf 11}, 198 (2017), 1709.07435.

\bibitem{Becchetti:2017abb}
M.~Becchetti and R.~Bonciani,
\newblock JHEP {\bf 01}, 048 (2018), 1712.02537.

\bibitem{Badger:2013gxa}
S.~Badger, H.~Frellesvig, and Y.~Zhang,
\newblock JHEP {\bf 12}, 045 (2013), 1310.1051.

\bibitem{Badger:2015lda}
S.~Badger, G.~Mogull, A.~Ochirov, and D.~O'Connell,
\newblock JHEP {\bf 10}, 064 (2015), 1507.08797.

\bibitem{Gehrmann:2015bfy}
T.~Gehrmann, J.~M. Henn, and N.~A. Lo~Presti,
\newblock Phys. Rev. Lett. {\bf 116}, 062001 (2016), 1511.05409,
\newblock [Erratum: Phys. Rev. Lett.116,no.18,189903(2016)].

\bibitem{Dunbar:2016aux}
D.~C. Dunbar and W.~B. Perkins,
\newblock Phys. Rev. {\bf D93}, 085029 (2016), 1603.07514.

\bibitem{Dunbar:2016cxp}
D.~C. Dunbar, G.~R. Jehu, and W.~B. Perkins,
\newblock Phys. Rev. {\bf D93}, 125006 (2016), 1604.06631.

\bibitem{Dunbar:2016gjb}
D.~C. Dunbar, G.~R. Jehu, and W.~B. Perkins,
\newblock Phys. Rev. Lett. {\bf 117}, 061602 (2016), 1605.06351.

\bibitem{Badger:2016ozq}
S.~Badger, G.~Mogull, and T.~Peraro,
\newblock JHEP {\bf 08}, 063 (2016), 1606.02244.

\bibitem{Dunbar:2017nfy}
D.~C. Dunbar, J.~H. Godwin, G.~R. Jehu, and W.~B. Perkins,
\newblock Phys. Rev. {\bf D96}, 116013 (2017), 1710.10071.

\bibitem{Badger:2017jhb}
S.~Badger, C.~Br{\o}nnum-Hansen, H.~B. Hartanto, and T.~Peraro,
\newblock Phys. Rev. Lett. {\bf 120}, 092001 (2018), 1712.02229.

\bibitem{Papadopoulos:2015jft}
C.~G. Papadopoulos, D.~Tommasini, and C.~Wever,
\newblock JHEP {\bf 04}, 078 (2016), 1511.09404.

\bibitem{Smirnov:2010hn}
A.~V. Smirnov and A.~V. Petukhov,
\newblock Lett. Math. Phys. {\bf 97}, 37 (2011), 1004.4199.

\bibitem{Laporta:2000dc}
S.~Laporta,
\newblock Phys. Lett. {\bf B504}, 188 (2001), hep-ph/0102032.

\bibitem{Laporta:2001dd}
S.~Laporta,
\newblock Int. J. Mod. Phys. {\bf A15}, 5087 (2000), hep-ph/0102033.

\bibitem{Anastasiou:2004vj}
C.~Anastasiou and A.~Lazopoulos,
\newblock JHEP {\bf 07}, 046 (2004), hep-ph/0404258.

\bibitem{Smirnov:2008iw}
A.~V. Smirnov,
\newblock JHEP {\bf 10}, 107 (2008), 0807.3243.

\bibitem{Smirnov:2014hma}
A.~V. Smirnov,
\newblock Comput. Phys. Commun. {\bf 189}, 182 (2014), 1408.2372.

\bibitem{Studerus:2009ye}
C.~Studerus,
\newblock Comput. Phys. Commun. {\bf 181}, 1293 (2010), 0912.2546.

\bibitem{vonManteuffel:2012np}
A.~von Manteuffel and C.~Studerus,
\newblock (2012), 1201.4330.

\bibitem{Lee:2012cn}
R.~N. Lee,
\newblock (2012), 1212.2685.

\bibitem{Maierhoefer:2017hyi}
P.~Maierh{\"o}fer, J.~Usovitsch, and P.~Uwer,
\newblock Comput. Phys. Commun. {\bf 230}, 99 (2018), 1705.05610.

\bibitem{Gluza:2010ws}
J.~Gluza, K.~Kajda, and D.~A. Kosower,
\newblock Phys.Rev. {\bf D83}, 045012 (2011), 1009.0472.

\bibitem{Ita:2015tya}
H.~Ita,
\newblock Phys. Rev. {\bf D94}, 116015 (2016), 1510.05626.

\bibitem{Larsen:2015ped}
K.~J. Larsen and Y.~Zhang,
\newblock Phys. Rev. {\bf D93}, 041701 (2016), 1511.01071.

\bibitem{vonManteuffel:2014ixa}
A.~von Manteuffel and R.~M. Schabinger,
\newblock Phys. Lett. {\bf B744}, 101 (2015), 1406.4513.

\bibitem{Kotikov:1990kg}
A.~V. Kotikov,
\newblock Phys. Lett. {\bf B254}, 158 (1991).

\bibitem{Kotikov:1991pm}
A.~V. Kotikov,
\newblock Phys. Lett. {\bf B267}, 123 (1991).

\bibitem{Bern:1993kr}
Z.~Bern, L.~J. Dixon, and D.~A. Kosower,
\newblock Nucl. Phys. {\bf B412}, 751 (1994), hep-ph/9306240.

\bibitem{Remiddi:1997ny}
E.~Remiddi,
\newblock Nuovo Cim. {\bf A110}, 1435 (1997), hep-th/9711188.

\bibitem{Gehrmann:1999as}
T.~Gehrmann and E.~Remiddi,
\newblock Nucl. Phys. {\bf B580}, 485 (2000), hep-ph/9912329.

\bibitem{Henn:2013pwa}
J.~M. Henn,
\newblock Phys. Rev. Lett. {\bf 110}, 251601 (2013), 1304.1806.

\bibitem{Papadopoulos:2014lla}
C.~G. Papadopoulos,
\newblock JHEP {\bf 07}, 088 (2014), 1401.6057.

\bibitem{Ablinger:2015tua}
J.~Ablinger {\em et~al.},
\newblock Comput. Phys. Commun. {\bf 202}, 33 (2016), 1509.08324.

\bibitem{Liu:2017jxz}
X.~Liu, Y.-Q. Ma, and C.-Y. Wang,
\newblock Phys. Lett. {\bf B779}, 353 (2018), 1711.09572.

\bibitem{Bern:2017gdk}
Z.~Bern, M.~Enciso, H.~Ita, and M.~Zeng,
\newblock Phys. Rev. {\bf D96}, 096017 (2017), 1709.06055.

\bibitem{Frellesvig:2017aai}
H.~Frellesvig and C.~G. Papadopoulos,
\newblock JHEP {\bf 04}, 083 (2017), 1701.07356.

\bibitem{Zeng:2017ipr}
M.~Zeng,
\newblock JHEP {\bf 06}, 121 (2017), 1702.02355.

\bibitem{Harley:2017qut}
M.~Harley, F.~Moriello, and R.~M. Schabinger,
\newblock JHEP {\bf 06}, 049 (2017), 1705.03478.

\bibitem{Moser:1959}
J.~Moser,
\newblock Mathematische Zeitschrift {\bf 72(1)}, 379 (1959).

\bibitem{Lee:2014ioa}
R.~N. Lee,
\newblock JHEP {\bf 04}, 108 (2015), 1411.0911.

\bibitem{Lee:2017oca}
R.~N. Lee and A.~A. Pomeransky,
\newblock (2017), 1707.07856.

\bibitem{Meyer:2016slj}
C.~Meyer,
\newblock JHEP {\bf 04}, 006 (2017), 1611.01087.

\bibitem{Prausa:2017ltv}
M.~Prausa,
\newblock Comput. Phys. Commun. {\bf 219}, 361 (2017), 1701.00725.

\bibitem{Gituliar:2017vzm}
O.~Gituliar and V.~Magerya,
\newblock Comput. Phys. Commun. {\bf 219}, 329 (2017), 1701.04269.

\bibitem{Meyer:2017joq}
C.~Meyer,
\newblock Comput. Phys. Commun. {\bf 222}, 295 (2018), 1705.06252.

\bibitem{Baikov:1996rk}
P.~A. Baikov,
\newblock Phys. Lett. {\bf B385}, 404 (1996), hep-ph/9603267.

\bibitem{Georgoudis:2016wff}
A.~Georgoudis, K.~J. Larsen, and Y.~Zhang,
\newblock Comput. Phys. Commun. {\bf 221}, 203 (2017), 1612.04252.

\bibitem{Bern:1992em}
Z.~Bern, L.~J. Dixon, and D.~A. Kosower,
\newblock Phys. Lett. {\bf B302}, 299 (1993), hep-ph/9212308,
\newblock [Erratum: Phys. Lett.B318,649(1993)].

\bibitem{Tarasov:1996br}
O.~V. Tarasov,
\newblock Phys. Rev. {\bf D54}, 6479 (1996), hep-th/9606018.

\bibitem{Lee:2009dh}
R.~N. Lee,
\newblock Nucl. Phys. {\bf B830}, 474 (2010), 0911.0252.

\bibitem{Lee:2010wea}
R.~N. Lee,
\newblock Nucl. Phys. Proc. Suppl. {\bf 205-206}, 135 (2010), 1007.2256.

\bibitem{Bosma:2017ens}
J.~Bosma, M.~S{\o}gaard, and Y.~Zhang,
\newblock JHEP {\bf 08}, 051 (2017), 1704.04255.

\end{thebibliography}

\end{document}